\documentclass[11pt]{article}
\usepackage{fullpage}
\usepackage{times}
\usepackage{mathptmx}
\usepackage{amsmath,amssymb,amsthm}
\usepackage{mathtools}
\usepackage{float}
\usepackage{pdfpages}
\usepackage[shortlabels]{enumitem}
 \setlist{nosep}
\usepackage[T1]{fontenc}
\usepackage{tikz}
\usepackage{pgfplots}
\pgfplotsset{
  label style={font=\footnotesize},
  tick label style={font=\footnotesize}
}
\newlength{\wdfig}
\newlength{\htfig}

\DeclareMathOperator{\diag}{diag}
\DeclareMathOperator{\grad}{\ensuremath{\mathrm{grad}}}
\DeclareMathOperator{\linspan}{\ensuremath{\mathrm{span}}}
\newcommand\norm[1]{\ensuremath{\lVert#1\rVert}}
\newcommand{\ball}[2]{\ensuremath{\mathcal{B}_{#1}(#2)}}

\newcommand{\Clem}{\ensuremath{\mathcal{C}_\text{lem}}}
\newcommand{\Csph}{\ensuremath{\mathcal{C}_\text{sph}}}
\newcommand{\G}{\ensuremath{\mathcal{G}}}
\newcommand{\V}{\ensuremath{\mathcal{V}}}
\newcommand{\Vu}{\ensuremath{\mathcal{V}_\text{u}}}
\newcommand{\Vl}{\ensuremath{\mathcal{V}_\text{l}}}
\newcommand{\E}{\ensuremath{\mathcal{E}}}
\newcommand{\Hn}{\ensuremath{\mathcal{H}^n}}
\newcommand{\Iu}{\ensuremath{\mathcal{I}_\text{u}}}
\newcommand{\Il}{\ensuremath{\mathcal{I}_\text{l}}}
\newcommand{\M}{\ensuremath{\mathcal{M}}}
\newcommand{\nei}{\ensuremath{\mathcal{N}}}
\newcommand{\normal}{\mathrm{n}}
\newcommand{\proj}{\mathrm{P}}
\newcommand{\tansp}[2][]{\ensuremath{\mathsf{T}_{#2}#1}}
\newcommand{\xu}{\ensuremath{x_\text{u}}}
\newcommand{\xel}{\ensuremath{x_\text{l}}}
\newcommand{\lowinf}{\mathop{\mathrm{inf}\vphantom{\mathrm{sup}}}}

\newtheorem{theorem}{Theorem}
\newtheorem{proposition}[theorem]{Proposition}
\newtheorem{corollary}[theorem]{Corollary}
\newtheorem{assumption}[theorem]{Assumption}
\newtheorem{definition}[theorem]{Definition}
\newtheorem{remark}[theorem]{Remark}

\title{Asymptotically stable polarization of multi-agent gradient flows over manifolds}
\author{La Mi \and Jorge Gon\c{c}alves \and Johan Markdahl}
\date{Luxembourg Centre for Systems Biomedicine\\
  University of Luxembourg\\
  \{la.mi, jorge.goncalves\}@uni.lu, markdahl@kth.se}
\begin{document}
\maketitle
\textbf{Abstract}.
Multi-agent systems are known to exhibit stable emergent behaviors, including polarization,
over \(\mathbb{R}^n\) or highly symmetric nonlinear spaces.
In this article, we eschew linearity and symmetry of the underlying spaces,
and study the stability of polarized equilibria of multi-agent gradient flows evolving on general hypermanifolds.
The agents attract or repel each other according to the partition of the communication graph that is connected but otherwise arbitrary.
The manifolds are outfitted with geometric features styled ``dimples'' and ``pimples'' that characterize the absence of flatness.
The signs of inter-agent couplings together with these geometric features give rise to stable polarization
under various sufficient conditions.
We propose tangible interpretation of the system in the context of opinion dynamics,
and highlight throughout the text its versatility in modeling various aspects of the polarization phenomenon.

\section{Introduction}
\label{sec:intro}
Let there be a collection of simple agents traversing a complex terrain.
Their goal is to approach the weighted average positions of their neighbors,
which they are programmed to either attract or repel.
This is a typical distributed process \cite{tron12riemannian},
and is akin to many aspects of social opinion formation,
as numerous observers in the systems and control community have identified \cite{proskurnikov17tutorial,proskurnikov18tutorial}.
Accordingly, insights on multi-agent systems, especially on the stability and convergence properties of consensus,
have been aptly applied to a dizzying array of opinion dynamics models to study social agreement
\cite{cao10convergence,mirtabatabaei12opinion,rossi20opinion}.
Continuing this tradition, we investigate multi-agent systems that evolve over nonlinear spaces
for another type of emergent behavior, namely, polarization.
Through this angle, we seek to understand the social phenomenon of opinion polarization,
which has become increasingly conspicuous in the backdrop of perceived deepening of social discords \cite{gentzkow19measuring}.

The specific empirical observations of the multi-faceted polarization phenomenon we are concerned with are the following.
At the individual level, the personal belief system of an agent is simpler and cleaner
compared to the richness of the external environment,
composed of ideas, events, and issues in the public sphere that are often unrelated to each other.
At the macroscopic level, polarization not only arises in a straightforward fashion
when two parties holding diametric views are squarely opposed to each other.
It might also happen among groups sharing significant common grounds with ostensibly minor differences
(one iota of difference between the Homoousians and the Homoiousians \cite[Chp.~21]{gibbon89decline}),
contrary to outsider intuition.
In addition, polarization is not merely characterized by clustering of opinions,
but can also be accompanied by radicalization, that is, the widening of differences between opinion clusters.
We show how our results capture these patterns that have not been addressed to date,
thus rendering our polarization problem over nonlinear space a suitable phenomenological model for opinion dynamics.
More concretely,
we study a model of multi-agent gradient flow system confined to manifolds embedded in the Euclidean space.
Gradient descent flow, being a sufficiently simple optimizing process, is amenable to rigorous stability analysis
and thus widely adopted by many agent-based models as coordinating protocols for robot swarms \cite{sepulchre11consensus}.
The restriction to nonlinear spaces in social dynamics analysis \cite{caponigro15nonlinear}, however, is less common,
as the majority of opinion models live in linear spaces.\footnote{
  We would like to note that not all works on opinion dynamics fall under this dichotomy.
  Notably, a sheaf-theoretic perspective is offered in \cite{hansen21opinion}
  where more complex data structures (rather than traditional vector valued) over general topological spaces
  may accommodate more sophisticated mechanisms in social discourses.}
They are typically contractive, in the sense that opinions become less extreme
irrespective of consensus or social cleavage as the final state,
even when explicit mechanisms to foster social cleavage is incorporated
such as antagonistic interactions \cite{altafini12consensus}, bounded confidence \cite{hegselmann02opinion},
and stubborn agents \cite{parsegov16novel,amelkin17polar}. %
On the other hand, we are motivated by the observation that the variety of issues or situations that an individual confronts is necessarily
more complex and nuanced than the set of a few principles, or core beliefs, used to navigate those situations.
The external events or environment may thus be naturally presented in the ambient Euclidean space,
and the core beliefs of an individual are encoded in the parametrization of a lower dimensional manifold.
The analytical investigation into this conceptual interpretation is made possible
by outfitting the general manifolds with special geometric features styled ``dimples'' and ``pimples''.
The interplay between these geometric features and cooperative/antagonistic interactions among agents
then gives rise to different routes to polarization.

There are a few polarization studies on manifolds in the literature, where the \(n\)-sphere has received the most attention.
Gaitonde et al \cite{gaitonde21polarization} studied a class of Markov processes,
where both the agent positions and random external stimuli have the same dimension.
The antipodal configuration is a natural definition for polarization,
and the geometric properties of the hypersphere are exploited to prove almost sure convergence.
Hong and Strogatz found traveling wave polarization in addition to the stationary type
in a variant of the Kuramoto model over the unit circle with conformist and contrarian oscillators \cite{hong11kuramoto}.
The bifurcation points between different steady states were solved exactly
by a series of reduction techniques applied to the mean field approximation.
A higher dimensional Kuramoto model analyzed by Ha et al \cite{ha20emergence}
also features positive and negative couplings between agents.
They obtained stability conditions on initial conditions, relative strengths of the two types of couplings,
and frequency matrices that govern the self dynamics.
More elaborate state-dependent interaction rules inspired by neuroscience are considered by Crnki{\'c} and Ja{\'c}imovi{\'c}
over a 3-sphere through a quaternion formulation \cite{crnkic18swarms}.
The antipodal configuration is asymptotically stable if agents attract or repulse each other when they are respectively close or far.
For ring graphs over the 2-sphere,
Song et al \cite{song17intrinsic} obtained asymptotically stable polarization with even number of agents.
Moreover, the result is almost global if the graph is undirected.
For more general manifolds, a recent work by Aydogdu et al \cite{aydo17opinion} explored geodesic and chordal interactions
between agents on general Riemannian manifolds,
and established the existence of various equilibria and orbits but without stability analysis.

In view of these related works,
our contribution is that we provide rigorous stability analysis of polarized equilibria for the multi-agent gradient descent system
with arbitrary connected network topology over more general manifolds.
When reduced to the hypersphere case,
our results generalize and complement the existing ones obtained in \cite{ha20emergence,song17intrinsic} (or in the literature).
Moreover, we draw on our unique interpretation for embedded lower dimensional manifolds in the context of opinion dynamics,
to address aspects of the opinion polarization hitherto unaccounted for.

More broadly, many multi-agent models over manifolds have found broad application in engineering, biology, and social science.
One of the most remarkable is the Kuramoto model of coupled phase oscillators over the unit circle,
which has been used to model synchronization problems such as firefly flashing \cite{buck68mechanism,ermentrout91adaptive},
circadian rhythms \cite{antonsen08external,childs08stability}, large scale power grids \cite{dorfler13synchronization}, and more.
The Lohe model on the \(n\)-sphere as a higher dimensional generalization of the Kuramoto model
can be applied in quantum synchronization on the Bloch sphere \cite{lohe10quantum}.
Moreover, many engineering problems concerning orientation alignment
are naturally considered on special orthogonal groups \cite{sarlette10coordinated}.
Accordingly, we expect our study on polarization over more general manifolds
to be relevant in multiple problems across disciplines, eg cellular division of labor shaped by organism geometry \cite{staps20evolution},
spacecraft coordination in the presence of space-time potential well, and beyond.

\section{Setup}
\label{sec:setup}
\subsection{Geometric features of the manifold}
\label{sec:hni}
Consider a closed and orientable hypersurface embedded in the Euclidean ambient space
\[
  \mathcal{H}^n = \{y \in \mathbb{R}^{n+1} \,|\, c(y) = 0\}
\]
implicitly characterized by a smooth \(C^2\) function \(c:\mathbb{R}^{n+1} \rightarrow \mathbb{R}\). %
The hypersurface \(\mathcal{H}^n\) separates its complement \(\mathbb{R}^{n+1} - \mathcal{H}^n\) into two disjoint sets,
one where \(c\) is positive and the other where \(c\) is negative~\cite{lima88jordan}. 
Without loss of generality, we identify the former with the unbounded set outside \(\mathcal{H}^n\),
and the latter with the bounded set inside \(\mathcal{H}^n\).
We adopt the convention that the unit normal \(\normal(x) = \nabla c(x) / \norm{\nabla c(x)}\) is outward-pointing
(i.e., pointing towards the unbounded set),
where we assume that the gradient in \(\mathbb{R}^{n+1}\) satisfies \(\nabla c(x) \neq 0\) for every \(x \in \mathcal{H}^n\).
Under this assumption, it can be shown that \(\Hn\) is a hypermanifold by the implicit function theorem.

The hypermanifold \(\Hn\) is equipped with special features: dimples and pimples.
To define them, introduce a height function \(h_x \colon \mathcal{H}^n \rightarrow \mathbb{R}\) with respect to a fixed \(x\)
\[
  h_x(y) \coloneqq \langle \normal(x),y \rangle, \quad \forall y \in \mathcal{H}^n.
\]
The height function gives the altitude of a point y along the axis spanned by \(\normal(x)\).
For the example of a 2-sphere, we can take the north pole \((0,0,1)\) as the fixed point.
Then the normal at the north pole is \((0,0,1)\).
Consequently, the height function of any point on the sphere gives its \(z\)-coordinate.
For notational convenience, if the fixed \(x\) carries a subscript, e.g.\ \(x_i\), then \(h_{x_i}\) is shortened as \(h_i\).
Similarly, \(\normal(x_i)\) is often shortened to \(\normal_i\).
Now we are ready to introduce the definitions for a dimple and a pimple.
\begin{definition}
  \label{def:imple}
  If for some \(x \in \mathcal{H}^n\),
  \(y = x\) is a strict local minimizer of \(h_x(y)\) in a sufficiently small neighborhood
  \(\mathcal{I}_x = \{y \in \mathcal{H}^n \,|\, \norm{y - x} < \epsilon\}\),
  then \(\mathcal{I}_x\) is referred to as a dimple, and \(x\) the bottom of the dimple.
  Similarly, if for some \(x \in \mathcal{H}^n\),
  \(y = x\) is a strict local maximizer of \(h_x(y)\) in a sufficiently small neighborhood
  \(\mathcal{I}_x\), then \(\mathcal{I}_x\) is referred to as a pimple, and \(x\) the bottom of the pimple.
\end{definition}
Figure~\ref{fig:fruity} illustrates the concepts of dimples and pimples in \(\mathbb{R}^3\) by three fundamental fruits.
\begin{figure}
  \centering
  \includegraphics[width=.3\linewidth]{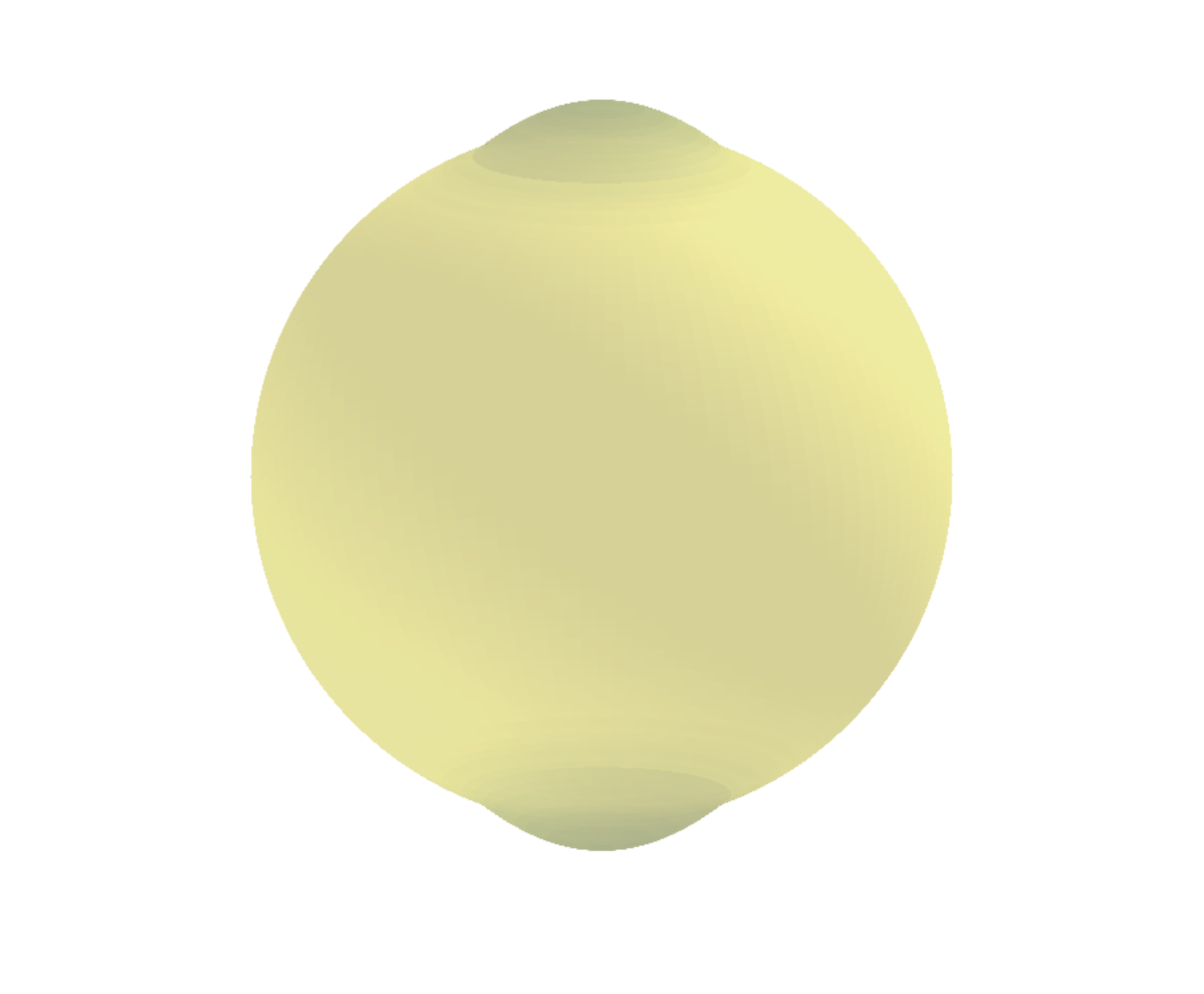}
  \includegraphics[width=.36\textwidth]{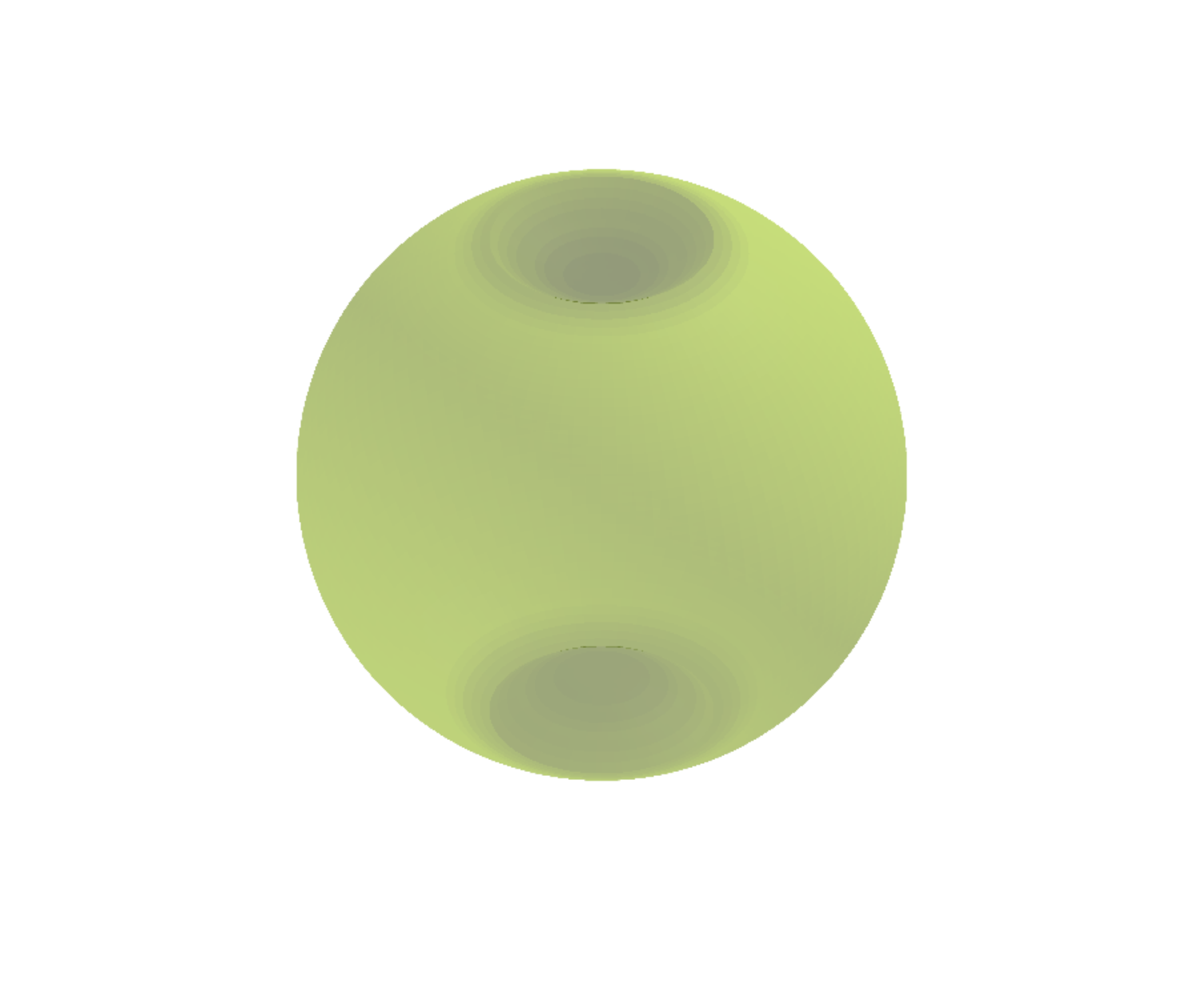}
  \includegraphics[width=.3\linewidth]{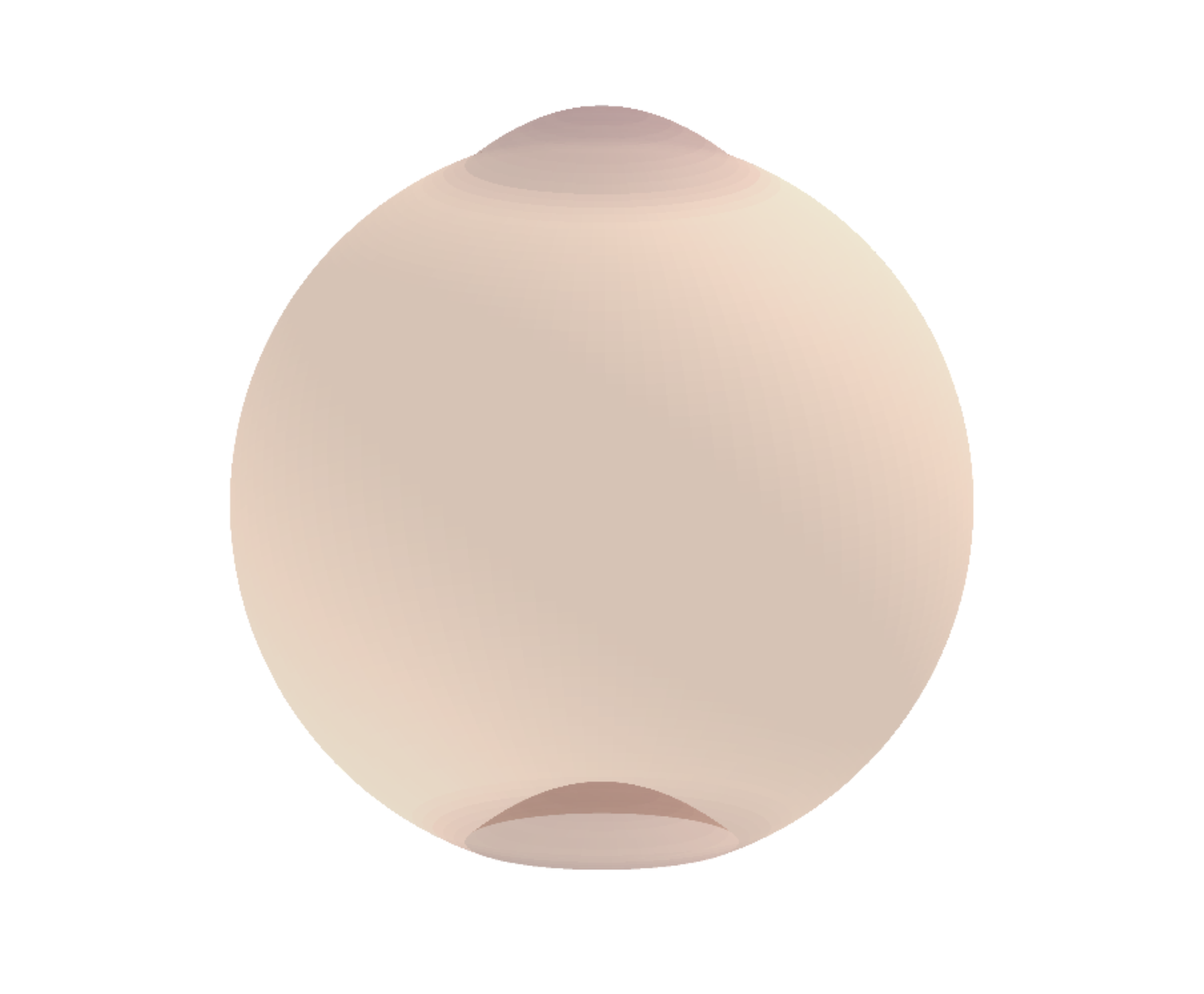}
  \caption{Three fundamental fruits illustrating Definition~\ref{def:imple}: a lemon with a pair of pimples (left),
    an apple with a pair of dimples (middle), and a peach with one pimple and one dimple (right).}
  \label{fig:fruity}
\end{figure}
\begin{remark}
  The dimple or pimple \(\mathcal{I}_x\) does not necessarily contain only one bottom \(x\).
  It may be that \(\mathcal{I}_x\) has a single set of bottoms covering all or part of \(\mathcal{I}_x\),
  or there are multiple disjoint sets of bottoms within \(\mathcal{I}_x\).
  We require \(\mathcal{I}_x\) to be sufficiently small in Definition \ref{def:imple} to exclude the latter case,
  by shrinking the radius \(\epsilon\) of the neighborhood around \(x\).
  However, it is not possible to avoid disjoint sets of bottoms when, e.g.,
  the embedding of the hypermanifold is not analytic.
\end{remark}

\subsection{Multi-agent networks}
\label{sec:gve}
Evolving on the hypermanifold is a homogeneous multi-agent system with \(N\) agents,
associated with an undirected, connected, and weighted graph structure \(\G = \left(\V, \E, A \right)\).
The adjacency matrix \(A = [a_{ij}]\) is symmetrical and nonnegative.
The vertices \(\V\) are divided into two groups \(\V_\text{u} = \{1,2,\dotso M\}\)
and \(\V_\text{l} = \{M+1,\dotso N\}\) for \(1 < M < N\).
The edge set \(\E\) is partitioned into intragroup and intergroup sets
\(\E_+ = \{\{i,j\} \in \E \,\vert\, i,j \in \V_\text{u} \text{ or } i,j \in \V_\text{l}\}\) and
\(\E_- = \{\{i,j\} \in \E \,\vert\, i \in \V_\text{l}, j \in \V_\text{u}\}\).

Such a partition is introduced to enforce different coupling rules over edges in \(\E_+\) and \(\E_-\).
The couplings are positive over all edges in \(\E_+\), whereas those over \(\E_-\) can be either all positive or all negative.
This ``edge coloring'' equivalently generates a \emph{structurally balanced} graph~\cite{altafini12consensus}
if we allow the graph to be signed such that edge weights on elements in \(\E_+\) are all positive, and
edge weights on elements in \(\E_-\) are either all positive or all negative.
In fact, doing so would not affect any of our results conceptually, as such,
no loss of generality is incurred with the non-negativity requirement on \(A\).

\begin{remark}
  It is noted in~\cite[Lem.~1]{altafini12consensus} that there exists a gauge transformation
  that brings a structurally balanced signed graph to a nonnegative one.
  This operation is akin to simultaneously reshuffling group membership
  and assigning to the affected agents opposite positions in the Euclidean space.
  It is not applicable in our general case (except the sphere case of special interest in §\ref{sec:sph}),
  because no symmetry is assumed in the underlying nonlinear space, as detailed in §\ref{sec:hni}.
  The symmetry assumption on \(A\), however, is essential, as we deal with gradient flows, see §\ref{sec:grad}.
\end{remark}

\subsection{Gradient flow dynamics}
\label{sec:grad}
Let us denote the states of the agents individually by \(x_i\) and collectively by \(\chi \coloneqq (x_i)_{i=1}^N\).
The agents evolve according to a simple rule of gradient descent flow in continuous time.
Given a disagreement function \(V \colon \Hn \rightarrow \mathbb{R}\), the dynamics of each agent is
\begin{equation}
  \label{eq:grad}
  \dot{x}_i = -\grad_i V(\chi) = -\proj_i\left(\nabla_i V(\chi)\right),
\end{equation}
where \(\grad_i\) is the intrinsic gradient on the tangent space at the point \(x_i\), denoted \(\tansp{x_i}{\mathcal{H}^n}\);
\(\proj_i = I - \normal_i \normal_i'\) is a positive semidefinite projection matrix \cite{markdahl20consensus}
on \(\tansp{x_i}{\mathcal{H}^n}\).

As mentioned in §\ref{sec:gve}, we divide the \(N\) agents into two groups
so that members within the same group are attracted to each other,
whereas members in different groups can be made to either oppose or attract each other.
To model the situation with attractive intragroup coupling and repulsive intergroup coupling on \(\mathcal{H}^n\),
we use the disagreement function
\begin{equation}
  \label{eq:disagree}
  V_-(\chi) \coloneqq \frac{1}{2} \sum_{\{i,j\} \in \E_+} a_{ij} \norm{x_j - x_i}_2^2
  - \frac{1}{2} \sum_{\{i,j\} \in \E_-} a_{ij} \norm{x_j - x_i}_2^2.
\end{equation}
Another appealing option in the literature is to retain the sum of square structure \cite{altafini12consensus} by changing the second term
\begin{equation}
  \label{eq:disagreealtafini}
  \frac{1}{2} \sum_{\{i,j\} \in \E_+} a_{ij} \norm{x_j - x_i}_2^2
  + \frac{1}{2} \sum_{\{i,j\} \in \E_-} a_{ij} \norm{x_j + x_i}_2^2.
\end{equation}
However, it is not applicable in our case in the absence of any symmetry in the nonlinear space.
Take a two particle system evolving on the unit circle as an example, every antipodal formation minimizes \eqref{eq:disagreealtafini}
only if the unit circle is centered at the origin.
Right shift by one unit, and consensus at the origin becomes the new minimizer.
To be coordinate agnostic, we choose \eqref{eq:disagree} to arrive at
\begin{equation}
  \label{eq:main}
  \begin{split}
    \dot{x}_i &= \left(I - \normal_i \normal_i'\right) \left(
      \sum_{j\in\V_\text{u}} a_{ij} \left(x_j - x_i\right) -
      \sum_{j\in\V_\text{l}} a_{ij} \left(x_j - x_i\right)
    \right) \quad \forall i\in\V_\text{u} \\
    \dot{x}_i &= \left(I - \normal_i \normal_i'\right) \left(
      \sum_{j\in\V_\text{l}} a_{ij} \left(x_j - x_i\right) -
      \sum_{j\in\V_\text{u}} a_{ij} \left(x_j - x_i\right)
    \right) \quad \forall i\in\V_\text{l}.
  \end{split}
\end{equation}
The sphere model with homogeneous frequency matrices \cite[(2.3)]{ha20emergence} is a special case of \eqref{eq:main}.
For the situation with purely attractive coupling, we simply change the sign of the second term in~\eqref{eq:disagree}:
\begin{equation}
  \label{eq:disagreepp}
  V_+(\chi) = \frac{1}{2} \sum_{\{i,j\} \in \E} a_{ij} \norm{x_j - x_i}_2^2,
\end{equation}
following which the state equation for agent \(i\) becomes
\begin{equation}
  \label{eq:mainpp}
  \dot{x}_i = \left(I - \normal_i \normal_i'\right) \sum_{j \in \V} a_{ij} \left(x_j - x_i\right) \quad \forall i \in \V.
\end{equation}
Or more compactly, \(\dot{\chi} = -\proj L \otimes I \chi\), where \(\proj = \diag(\proj_i)\) and \(L\) is the familiar graph Laplacian.
It is readily recognizable as a nonlinear higher-dimensional version of the Abelson model \cite{abelson67mathematical},
where the nonlinearity comes from the underlying nonlinear core belief space.
Moreover, the sphere version of \eqref{eq:mainpp} is a Kuramoto model with homogeneous frequencies
\cite{hemmen93lyapunov, markdahl21almost}.

\subsection{The assemblage}
\label{sec:prel}
Assembling the aforementioned ingredients in §\ref{sec:gve}-\ref{sec:grad},
we have a multi-agent gradient flow system with attractive~\eqref{eq:main} or repulsive~\eqref{eq:mainpp} intergroup interactions
evolving on a hypermanifold \(\mathcal{H}^n\).
The hypermanifold is equipped with a pair of dimples or pimples as illustrated in Fig.~\ref{fig:fruity},
each containing one of the two groups of agents \(\Vu\) and \(\Vl\).
We are interested in possible polarization arising in this setting
as a result of the interplay between the graph couplings and the geometry of the underlying nonlinear space.
\begin{definition}[polarization]
  \label{def:pol}
  The agents are said to be polarized if \(x_i = x_j\) for all \(\{i,j\} \in \E_+\)
  and \(x_i \neq x_j\) for all \(\{i,j\} \in \E_-\).
\end{definition}
This definition is less in common with the notion of bipartite consensus over linear Euclidean space in \cite{altafini12consensus}
meaning all agents are equal up to signs,
or with the strong/weak polarization \cite{gaitonde21polarization} and complete polarization \cite{ha20emergence}
over the hypersphere which are antipodal.
Those definitions are convenient for highly symmetric spaces,
whereas we work with more general hypermanifolds without symmetry and cannot define polarization this way.
However, our definition enjoys more flexibility,
as we can model opposition over seemingly trivial differences (from the external point of view)
with Definition~\ref{def:pol} plus \(\norm{x_i-x_j} < \epsilon\) for \(\{i,j\}\in\E_-\).
In this sense, the partial polarization defined in \cite{ha20emergence} is more closely related,
where the two groups on the sphere need not be antipodal and is possible in their case due to the presence of frequency terms.

Definition~\ref{def:pol} characterizes a polarized configuration
without specifying whether the states are in equilibrium, limit-cycle, or other non-stationary modes.
We focus on polarized equilibria and their stability properties,
which are the only possible final outcome as shown by the following proposition.
Let \(u_i\) denote the aggregate attraction for \(x_i\)
\begin{equation}
  \label{eq:uiplus}
  u_i = \sum_{j \in \V} a_{ij} \left(x_j - x_i\right).
\end{equation}
Analogous to~\cite[Prop.~11]{markdahl18almost}, we have the following result.
\begin{proposition}
  \label{prop:isequilib}
  The system~\eqref{eq:mainpp} converges to an equilibrium set on \(\left(\mathcal{H}^n\right)^N\).
  At the equilibrium, \(u_i\) defined in \eqref{eq:uiplus} is either zero or parallel to \(\normal_i\).
\end{proposition}
\begin{proof}
  For \(V_+(\chi)\) in~\eqref{eq:disagreepp},
  \[
    \begin{split}
      \dot{V}_+ &= \sum_{\{i,j\} \in \E} a_{ij} \left(\langle{x_j-x_i, \dot{x}_j}\rangle - \langle{x_j-x_i, \dot{x}_i}\rangle\right) \\
      &= \sum_{i \in \V} \langle{\dot{x}_i, \sum_{j\in\V} a_{ij}(x_i - x_j)}\rangle +
      \sum_{j \in \V} \langle{\dot{x}_j, \sum_{i\in\V} a_{ij}(x_j - x_i)}\rangle \\
      &= -2 \sum_{i\in\V} \langle{\dot{x}_i, \sum_{j\in\V} a_{ij}(x_j - x_i)}\rangle \\
      &= -2 \sum_{i\in\V} \langle{\proj_i u_i, u_i}\rangle.
    \end{split}
  \]
  Therefore, we have \(\dot{V}_+ \leq 0\).
  By LaSalle's invariance principle, system~\eqref{eq:mainpp} converges to the set \(\{(x_i)_{i=1}^N \,|\, \dot{V}_+ = 0\}\).
  To achieve \(\dot{V}_+ = 0\), either \(u_i = 0\) or \(u_i \in \ker P_i\).
  In the latter case, as \(\ker P_i = \linspan{(\normal_i)}\), we have that \(u_i\) is parallel to \(\normal_i\).
  Inspecting~\eqref{eq:mainpp}, we see that both cases lead to \(\dot{x} = 0\) for all \(i \in \V\).
  Therefore, the system~\eqref{eq:mainpp} converges to an equilibrium set on \(\left(\mathcal{H}^n\right)^N\).
\end{proof}
As noted in §\ref{sec:gve}, we may use signed graphs for intergroup antagonism,
then \eqref{eq:disagree} and \eqref{eq:main} would respectively reduce in form to \eqref{eq:disagreepp} and \eqref{eq:mainpp}.
And the same conclusion as Prop.~\ref{prop:isequilib} can be reached via identical arguments for system \eqref{eq:main} with \(V_-\).
\begin{remark}
  We emphasize that Prop.~\ref{prop:isequilib} dictates convergence to a set, rather than to a point of equilibrium.
  Hence, it does not exclude non-stationary steady states
  such as traveling waves \cite{hong11kuramoto} and dancing equilibria \cite{aydo17opinion}.
\end{remark}

For the definitions of Lyapunov and asymptotic stability, while those of an equilibrium point are well known,
those of a set of equilibria is perhaps less standard.
Introduce the \emph{Hausdorff distance} between two sets \(\mathcal{Y}, \mathcal{Z} \subset \mathbb{R}^n\),
\[
  d_\text{H}(\mathcal{Y}, \mathcal{Z}) \coloneqq
  \max\{\sup_{y\in\mathcal{Y}} \lowinf_{z\in\mathcal{Z}} \norm{y-z}, \sup_{z\in\mathcal{Z}} \lowinf_{y\in\mathcal{Y}} \norm{y-z}\}.
\]
\begin{definition}[stability]
  A set of equilibria \(\mathcal{S}\) is Lyapunov stable if, for each \(\epsilon > 0\), there is \(\delta = \delta(\epsilon)\)
  such that \(d_\text{H}(x, \mathcal{S})|_{t=0} < \delta\)
  implies \(d_\text{H}(x, \mathcal{S})|_t < \epsilon\) for all \(t \geq 0\);
  is asymptotically stable if it is stable and \(\delta\) can be chosen
  such that \(d_\text{H}(x, \mathcal{S})|_{t=0} < \delta\) implies
  \(\lim_{t \rightarrow \infty} d_\text{H}(x, \mathcal{S}) = 0\).
\end{definition}
We collect a few previous results and associated definitions that will pave the way for the development of our main results.
\begin{definition}[local minimizer] %
  A set \(\mathcal{S} \subset \mathcal{X}\) of minimizers of a real function \(f \colon \mathcal{X} \rightarrow \mathbb{R}\)
  from a metric space \((\mathcal{X}, d_\text{H})\) is said to be a \emph{local minimizer} if for some \(\epsilon > 0\)
  there is an open neighborhood \(\mathcal{N}(\mathcal{S}) = \{x \in \mathcal{M} \,|\, d_\text{H}(x, \mathcal{S}) < \epsilon\}\)
  such that \(f|_\mathcal{S} \leq f(x)\) for all \(x \in \mathcal{N}(\mathcal{S})\).
  Moreover, if the inequality is strict for all \(x \in \mathcal{N}(\mathcal{S}) \backslash \mathcal{S}\),
  then \(\mathcal{S}\) is said to be a \emph{strict local minimizer}.
\end{definition}
\begin{definition}[isolated critical] %
  \label{def:isocri}
  A set \(\mathcal{S} \subset \mathcal{X}\) of critical points of a real function \(f \colon \mathcal{X} \rightarrow \mathbb{R}\)
  from a metric space \((\mathcal{X}, d_\text{H})\) is said to be \emph{isolated critical} if for some \(\epsilon > 0\)
  there is an open neighborhood \(\mathcal{N}(\mathcal{S}) = \{x \in \mathcal{M} \,|\, d_\text{H}(x, \mathcal{S}) < \epsilon\}\)
  such that \(\mathcal{N}(\mathcal{S}) \backslash \mathcal{S}\) is void of critical points.
\end{definition}
\begin{proposition}[Prop. 7~\cite{markdahl20synchronization}]
  \label{prop:lyapunov}
  Let \(\mathcal{M}\) be closed and take any \(V \colon \mathcal{M} \rightarrow \mathbb{R}\) that is \(C^2\).
  Let \(\mathcal{S}\) be a compact set of local minimizers of \(V\).
  If \(\mathcal{S}\) is a strict local minimizer,
  then \(\mathcal{S}\) is a Lyapunov stable equilibrium set of \(\dot{x} = -\grad V\).
  If \(\mathcal{S}\) is also isolated critical, then it is asymptotically stable.
\end{proposition}

\section{Fruits}
\label{sec:fruits}
In this section, we present and discuss our main results concerning the stability properties of polarized equilibria.
They arise in different combinations of attractive/repulsive interactions with dimple/pimple geometric features,
best exemplified by the fundamental fruits portrayed in Fig.~\ref{fig:fruity}.
Despite the symmetrical shapes that the fruits assume in the figure,
we emphasize that our general results do not require any spatial symmetry of the hypermanifold embedding.
\subsection{Apple}
\label{sec:dimapp}
Let us consider the setting of a pair of dimples on the hypersurface,
one containing the group \(\V_\text{u}\) and the other \(\V_\text{l}\).
Thus, we operate under the following assumption in this section:
\begin{assumption}
  \label{ass:dim}
  The sets \(\mathcal{I}_\text{u}\) and \(\mathcal{I}_\text{l}\) are a pair of dimples, and
  \(x_i \in \mathcal{I}_\text{u}\) for all \(i \in \V_\text{u}\), \(x_i \in \mathcal{I}_\text{l}\) for all \(i \in \V_\text{l}\).
\end{assumption}
Now we investigate stability of a polarized equilibrium that exists in the system~\eqref{eq:mainpp}
if the manifold resembles the apple in Fig.~\ref{fig:fruity}.
The point of interest is the following:
\begin{equation}
  \label{eq:capp}
  \begin{split}
    \chi^\ast &\coloneqq
    \{\chi \in (\mathcal{H}^n)^N \,|\, x_i = \xu\> \forall i \in \V_\text{u}, x_i = \xel\> \forall i \in \V_\text{l}\}.
  \end{split}
\end{equation}

Since the graph \(\G\) is connected, \(\E_-\) is non-empty.
Consequently, under polarization as per Definition~\ref{def:pol}, \(u_i\) defined in \eqref{eq:uiplus} is non-zero,
and hence must be parallel to \(\normal_i\) according to Prop.~\ref{prop:isequilib}.
This observation motivates the conditions set forth in the ensuing results.
\begin{proposition}
  \label{clm:dimhmin}
  For system~\eqref{eq:mainpp} under Assumption~\ref{ass:dim},
  if there exists a pair of distinct dimple bottoms
  \(x_\text{u} \in \mathcal{I}_\text{u}\) and \(x_\text{l} \in \mathcal{I}_\text{l}\) such that
  \begin{enumerate}[(i)]
  \item \(x_\text{u} - x_\text{l}\) is parallel to \(\normal(x_\text{u})\), and
  \item \(h_\text{u}(x_\text{l})\) is a local maximum satisfying \(h_\text{u}(x_\text{l}) < h_\text{u}(x_\text{u})\),
\end{enumerate}

  then a strict local minimum of \(V_+\) is
  \(V_+^\ast \coloneqq \frac{1}{2} (h_\text{u}(x_\text{u}) - h_\text{u}(x_\text{l}))^2 \sum_{\{i,j\} \in \E_-} a_{ij}\),
  and the corresponding strict local minimizer is \(\chi^\ast\) defined in~\eqref{eq:capp}.
\end{proposition}
\begin{proof}
  For all \(\{i,j\} \in \E_-\), assume without loss of generality that \(i \in \V_\text{u}\) and \(j \in \V_\text{l}\).
  The term \(\norm{x_j - x_i}_2^2\) in~\eqref{eq:disagreepp} is lower bounded by
  \begin{equation}
    \label{eq:pf:dd}
    \norm{x_j - x_i}_2^2 \geq \langle x_j - x_i, \normal(x_\text{u}) \rangle^2
    = \left(h_\text{u}(x_j) - h_\text{u}(x_i)\right)^2 \geq (h_\text{u}(x_\text{u}) - h_\text{u}(x_\text{l}))^2.
  \end{equation}
  The first lower bound is achieved only when \(x_j - x_i\) is parallel to \(\normal(\xu)\).
  For the second lower bound, observe that \(\xu\) is the bottom of the dimple \(\Iu\).
  Consequently, \(h_\text{u}(\xu)\) is by Definition~\ref{def:imple} a strict local minimum achieved only by \(\xu\).
  Combining this observation with condition (ii) in the proposition,
  we conclude that the second lower bound is achieved only when \(h_\text{u}(x_i) = h_\text{u}(\xu)\), implying \(x_i = \xu\),
  and \(h_\text{u}(x_j) = h_\text{u}(\xel)\).
  With condition (i) in the proposition, the only configuration to achieve both lower bounds
  is thus \(x_i\) and \(x_j\) being in their respective bottoms \(\xu\) and \(\xel\).

  For all \(\{i,j\} \in \E_+\),
  \[
    \norm{x_j - x_i}_2^2 \geq 0,
  \]
  where the equality is achieved when \(x_i = x_j\),
  of which a special case is \(x_i = x_j = x_\text{u}\) for \(i, j \in \V_\text{u}\)
  and \(x_i = x_j = x_\text{l}\) for \(i, j \in \V_\text{l}\).
  Therefore,
  \[
    V_+(\chi) \geq \frac{1}{2} \sum_{\{i,j\} \in \E_-} a_{ij} \norm{x_j - x_i}_2^2 \geq
    \frac{1}{2} (h_\text{u}(x_\text{u}) - h_\text{u}(x_\text{l}))^2 \sum_{\{i,j\} \in \E_-} a_{ij}.
  \]
  The minimum is achieved only when \(x_i = \xu, \forall i \in \V_\text{u}\) and \(x_i = \xel, \forall i \in \V_\text{l}\).
\end{proof}
\begin{corollary}
  \label{cor:dimhlyap}
  For system~\eqref{eq:mainpp} under Assumption~\ref{ass:dim},
  if the two dimples \(\mathcal{I}_\text{u}\) and \(\mathcal{I}_\text{l}\) satisfy the properties given in Prop.~\ref{clm:dimhmin},
  then \(\chi^\ast\) defined in~\eqref{eq:capp} is a Lyapunov stable polarized equilibrium.
\end{corollary}
\begin{proof}
  This is a direct application of Prop.~\ref{prop:lyapunov} to Prop.~\ref{clm:dimhmin}.
\end{proof}
Next, we derive an asymptotic stability result for when \(\Iu\) and \(\Il\) live on nice manifolds.
\begin{theorem}
  \label{clm:dimhasy}
  For system~\eqref{eq:mainpp} under Assumption~\ref{ass:dim},
  if the two dimples \(\Iu\) and \(\Il\) satisfy the properties given in Prop.~\ref{clm:dimhmin}, and in addition,
  there is a neighborhood \(\mathcal{N}_\text{a}(\chi^\ast)\) on \((\mathcal{H}^n)^N\) that belongs to an analytic manifold,
  then \(\chi^\ast\) defined in \eqref{eq:capp} is an asymptotically stable polarized equilibrium.
\end{theorem}
\begin{proof}
  Following a variant~\cite[Sec.~9]{kurdyka00proof} of the \L{}ojasiewicz inequality valid on analytic Riemannian manifolds,
  the analytic function \(V_+(\chi)\) in~\eqref{eq:disagreepp} behaves in the following way
  in a neighborhood of the polarized equilibrium \(\mathcal{N}_\text{\l}(\chi^\ast) \subset \nei_\text{a}(\chi^\ast)\):
  \[
    \lvert{V_+(\chi) - V_+(\chi^\ast)}\rvert^\alpha \leq \kappa \norm{\grad V_+(\chi)},
  \]
  for \(\alpha < 1\) and \(\kappa > 0\), and where \(\grad\) is the intrinsic gradient on the tangent space.
  For a Riemannian manifold that is also a hypersurface, \(\grad V_+(\chi) = \proj(\nabla V_+(\chi))\).
  Suppose that \(\chi \neq \chi^\ast\) is an equilibrium in \(\nei_\text{\l}(\chi^\ast)\),
  then \(\grad V_+(\chi) = \proj(\nabla V_+(\chi)) = 0\), c.f.~\eqref{eq:grad}.
  Consequently, \(V_+(\chi) = V_+(\chi^\ast)\).
  However, Prop.~\ref{clm:dimhmin} says that the local minimum \(V_+^\ast\) is achieved only if \(\chi = \chi^\ast\), a contradiction.
  Therefore, there is no equilibrium save for \(\chi^\ast\) in \(\nei_\text{\l}(\chi^\ast)\),
  rendering \(\chi^\ast\) isolated critical as per Definition~\ref{def:isocri}.
  Thus, \(\chi^\ast\) is asymptotically stable by Prop.~\ref{prop:lyapunov}.
\end{proof}
The results of this section suggest that even when two parties try to reach a common ground by making compromises,
thereby getting closer to each other in the Euclidean space,
the unbridgeable core belief landscape presents an obstruction to the eventual concord.
Especially interesting is when the Euclidean distance between the polarized points are small as seen from the extrinsic point of view,
it may represent an irreconcilable long lasting hostility between two groups with minor differences.

\subsection{Lemon}
\label{sec:pimlem}
Consider a pair of pimples on the hypersurface, one containing the group of agents \(\V_\text{u}\) and the other \(\V_\text{l}\).
The assumption in this section is then
\begin{assumption}
  \label{ass:pim}
  The sets \(\mathcal{I}_\text{u}\) and \(\mathcal{I}_\text{l}\) are a pair of pimples, and
  \(x_i \in \mathcal{I}_\text{u}\) for all \(i \in \V_\text{u}\), \(x_i \in \mathcal{I}_\text{l}\) for all \(i \in \V_\text{l}\).
\end{assumption}
For the dynamics, we are interested in \eqref{eq:main} with attractive intragroup coupling and repulsive intergroup coupling
corresponding to the disagreement function \eqref{eq:disagree}.
Thus, we may picture the system \eqref{eq:main} evolving on a lemon-like manifold (Fig.~\ref{fig:fruity}).

Denote a \emph{closed} ball centered at a point \(x\) with radius \(r\) as \(\mathcal{B}_r(x)\).
Let \(x_o = \frac{1}{2}(x_\text{u} + x_\text{l})\) denote the midpoint between \(x_\text{u}\) and \(x_\text{l}\),
and \(r_o = \frac{1}{2}\norm{x_\text{u} - x_\text{l}}\) the half distance between them.
The following results concern the set
\begin{equation}
  \label{eq:clem}
  \mathcal{C}_\text{lem} \coloneqq
  \{\chi \in (\mathcal{H}^n)^N \,|\, x_i = x\> \forall i \in \V_\text{u}, x_i = y\> \forall i \in \V_\text{l}, (x,y) \in Y\},
\end{equation}
where \(Y \coloneqq \{(x,y) \in \mathcal{I}_\text{u} \times \mathcal{I}_\text{l} \,|\, \norm{x-y}_2 = 2r_o\}\).
This set has at least one element \(\chi^\ast \in \Clem\).
\begin{proposition}
  \label{clm:pimballmin}
  For system~\eqref{eq:main} under Assumption~\ref{ass:pim}, %
  if there exists a pair of distinct pimple bottoms \(x_\text{u} \in \Iu\) and \(x_\text{l} \in \Il\) such that
  \(\mathcal{I}_\text{u}\) and \(\mathcal{I}_\text{l}\) %
  are entirely contained in \(\mathcal{B}_{r_o}(x_o)\),
  then a strict local minimum of \(V_-\) is \(V_-^\ast \coloneqq - 2 r_o^2 \sum_{\{i,j\} \in \E_-} a_{ij}\),
  and the corresponding strict local minimizer is a compact set of polarized configurations \(\Clem\) defined in~\eqref{eq:clem}.
\end{proposition}
\begin{proof}
  For all \(\{i,j\} \in \E_-\), assume without loss of generality that \(i \in \V_\text{u}\) and \(j \in \V_\text{l}\).
  The term \(\norm{x_j - x_i}_2^2\) in~\eqref{eq:disagree} is upper bounded by
  \[
    \norm{x_j - x_i}_2 \leq 2 r_o = \norm{x_\text{u} - x_\text{l}}_2 %
  \]
  The inequality is because both pimples are entirely contained in \(\mathcal{B}_{r_o}(x_o)\).
  The upper bound is achieved when \(x_i = x\) and \(x_j = y\)
  for every pair \((x \in \mathcal{I}_\text{u}, y \in \mathcal{I}_\text{l})\) such that \(\norm{x-y}_2 = 2r_o\),
  an example of which is \(x = x_\text{u}\) and \(y = x_\text{l}\).
  For all \(\{i,j\} \in \E_+\), the reasoning is identical to the corresponding part in the proof of Prop.~\ref{clm:dimhmin}.
  To repeat, we have
  \[
    \norm{x_j - x_i}_2^2 \geq 0,
  \]
  where the equality is achieved when \(x_i = x_j\).
  Therefore,
  \[
    V_-(\chi) \geq - \frac{1}{2} \sum_{\{i,j\} \in \E_-} a_{ij} \norm{x_j - x_i}_2^2 \geq
    - 2 r_o^2 \sum_{\{i,j\} \in \E_-} a_{ij}
  \]
  The minimum is achieved only when \(x_i = x, \forall i \in \V_\text{u}\) and \(x_i = y, \forall i \in \V_\text{l}\)
  for every pair of \((x \in \mathcal{I}_\text{u}, y \in \mathcal{I}_\text{l})\) such that \(\norm{x-y}_2 = 2r_o\).
\end{proof}
\begin{corollary}
  \label{cor:pimballyap}
  For system~\eqref{eq:main} under Assumption~\ref{ass:pim}, if the two pimples \(\mathcal{I}_\text{u}\) and \(\mathcal{I}_\text{l}\)
  satisfy the conditions given in Prop.~\ref{clm:pimballmin},
  then \(\mathcal{C}_\text{lem}\) defined in~\eqref{eq:clem} is a Lyapunov stable set of polarized equilibria.
\end{corollary}
\begin{proof}
  This is a direct application of Prop.~\ref{prop:lyapunov} to Prop.~\ref{clm:pimballmin}.
\end{proof}
\begin{theorem}
  \label{clm:pimballasy}
  For system~\eqref{eq:main} under Assumption~\ref{ass:pim}, if the two pimples \(\mathcal{I}_\text{u}\) and \(\mathcal{I}_\text{l}\)
  satisfy the conditions given in Prop.~\ref{clm:pimballmin},
  and in addition, there is a neighborhood \(\mathcal{N}_\text{a}(\mathcal{C}_\text{lem})\) on \((\mathcal{H}^n)^N\)
  that belongs to an analytic manifold,
  then \(\mathcal{C}_\text{lem}\) defined in~\eqref{eq:clem} is an asymptotically stable set of polarized equilibria.
\end{theorem}
Theorem~\ref{clm:pimballasy} concerns the asymptotic stability of an equilibrium set,
rather than an equilibrium point in Theorem~\ref{clm:dimhasy}.
The accompanying subtlety requires a more involved proof;
a similar line of reasoning appeared in the proof of \cite[Thm.~15]{markdahl21counterexamples}.
\begin{proof}
  The first part of the proof is to show that pointwise,
  there are no equilibrium points other than those in \(\Clem\) in the neighborhood of each \(\chi_\text{p} \in \Clem\).
  The arguments are identical to those in the proof of Theorem~\ref{clm:dimhasy} for the asymptotic stability of \(\chi^\ast\),
  except that we label quantities by a subscript ``p'' to stress that they vary with each different \(\chi_\text{p} \in \Clem\).
  Following a variant~\cite[Sec.~9]{kurdyka00proof} of the \L{}ojasiewicz inequality valid on analytic Riemannian manifolds,
  the analytic function \(V_-(\chi)\) in~\eqref{eq:disagree} behaves in the following way
  in a neighborhood \(\mathcal{N}_\text{\l}(\chi_\text{p}) \subset \nei_\text{a}(\Clem)\)
  of every polarized equilibrium \(\chi_\text{p} \in \mathcal{C}_\text{lem}\):
  \[
    \lvert{V_-(\chi) - V_-(\chi_\text{p})}\rvert^{\alpha_\text{p}} \leq \kappa_\text{p} \norm{\grad V_-(\chi)},
  \]
  for \(\alpha_\text{p} < 1\) and \(\kappa_\text{p} > 0\), and where \(\grad\) is the intrinsic gradient on the tangent space.
  For a Riemannian manifold that is also a hypersurface, \(\grad V_-(\chi) = \proj(\nabla V_-(\chi))\).
  Suppose that \(\chi \in \nei_\text{\l}(\chi_\text{p})\) is an element in the set of equilibria \(\mathcal{Q}\)
  such that \(\mathcal{Q} \cap \mathcal{C}_\text{lem} = \emptyset\),
  then \(\grad V_-(\chi) = \proj(\nabla V_-(\chi)) = 0\), c.f.~\eqref{eq:grad}.
  Consequently, \(V_-(\chi) = V_-(\chi_\text{p})\).
  However, Claim~\ref{clm:pimballmin} says that the local minimum \(V_-^\ast\) is achieved
  only if \(\chi \in \mathcal{C}_\text{lem}\), a contradiction.
  Therefore, \(\mathcal{Q} \cap \mathcal{N}_\text{\l}(\chi_\text{p}) = \emptyset, \forall \chi_p \in \mathcal{C}_\text{lem}\).

  Having shown that every point in \(\mathcal{C}_\text{lem}\) is isolated from \(\mathcal{Q}\),
  we proceed to demonstrate that no sequence in \(\mathcal{Q}\) can approach \(\mathcal{C}_\text{lem}\) arbitrarily close.
  Suppose on the contrary that there is such a sequence \(\{\chi_i\}_{i=1}^\infty \in \mathcal{Q}\),
  then \(\inf_{\chi \in \mathcal{Q}} V_-(\chi) = V_-^\ast\) and \(\lim_{i\rightarrow\infty} V_-(\chi_i) = V_-^\ast\).
  The sequence is bounded, since \(\mathcal{I}_\text{u}\) and \(\mathcal{I}_\text{l}\) are bounded sets.
  By the Bolzano-Weierstrass theorem, the sequence \(\{\chi_i\}_{i=1}^\infty\) has a subsequence
  that converges to some point \(\chi_q\) such that \(V(\chi_q) = V_-^\ast\),
  which in turn implies that \(\chi_q \in \mathcal{C}_\text{lem}\).
  This is to say that this subsequence in \(\mathcal{Q}\) must converge to a point \(\chi_q \in \mathcal{C}_\text{lem}\),
  contradicting the fact that \(\mathcal{Q} \cap \mathcal{N}_\text{\l}(\chi_\text{q}) = \emptyset\).

  Thus, we have shown that \(\mathcal{C}_\text{lem}\) is isolated critical,
  and therefore is asymptotically stable by Prop.~\ref{prop:lyapunov}.
\end{proof}
\begin{remark}
  Theorems~\ref{clm:dimhasy} and \ref{clm:pimballasy} for asymptotic stability rely on an analyticity assumption of the manifold.
  Although a smooth manifold may be topologically equivalent to an analytic one,
  it is worth stressing that our results concern manifolds with special geometric features,
  see Definition~\ref{def:imple}, and are thus not amenable to such topological equivalence.
\end{remark}
\begin{remark}
  The additional requirement on the analyticity of the manifold in Theorems~\ref{clm:dimhasy} and \ref{clm:pimballasy} is a local one.
  We do not require the whole manifold to be analytic
  for \(\mathcal{C}_\text{lem}\) or \(\chi^\ast\) to be asymptotically stable.
  For instance, \(\nei_\text{a}(\Clem)\) may be a subset of \((\Hn)^N \cap \M^N\),
  where \(\Hn\) is the hypermanifold on which the agents inhabit, whereas \(\M\) is an analytic manifold.
\end{remark}
The condition proposed in Prop.~\ref{clm:pimballmin} essentially seeks to ensure that \(2r_o\) is the largest possible distance
between all possible pairs of points \(\{x,y\} \in \Iu \times \Il\) in the neighborhood,
whereby \(V_-^\ast\) is a strict local minimizer of \(V_-\).
However, the conservatism introduced by this approach can be immediately identified.
Even if one pimple, say \(\Iu\), is outside \(\ball{r_o}{x_o}\), if the other pimple \(\Il\) is very ``steep'',
it may still be the case that \(2r_o\) is the largest possible distance
between all possible pairs of points \(\{x,y\} \in \Iu \times \Il\).
Nonetheless, if both pimples are outside \(\ball{r_o}{x_o}\), then instability can be established for \(\chi^\ast\),
as we show in the following.
\begin{proposition}
  \label{clm:pimballins}
  For system~\eqref{eq:main} under Assumption~\ref{ass:pim}, if the two pimples \(\mathcal{I}_\text{u}\) and \(\mathcal{I}_\text{l}\)
  are entirely outside \(\mathcal{B}_{r_o}(x_o)\) except for the bottoms \(\xu\) and \(\xel\),
  then \(\chi^\ast\) defined in~\eqref{eq:capp} is an equilibrium not asymptotically stable.
\end{proposition}
\begin{proof}
  That \(\chi^\ast\) is an equilibrium can be checked by substituting \eqref{eq:capp} into \eqref{eq:main}.
  To show asymptotic instability, we find a perturbation that causes the trajectory not to settle into \(\chi^\ast\).
  We choose a perturbation \(\tilde{\chi}\) that leaves the configuration polarized,
  that is \(x_i = \tilde{x}_\text{u} \neq \xu\) for all \(i \in \Vu\) and \(x_i = \tilde{x}_\text{l} \neq \xel\) for all \(i \in \Vl\),
  such that \(\tilde{x}_\text{u} - \tilde{x}_\text{l}\) passes through \(x_o\).
  As both pimples are outside \(\ball{r_o}{x_o}\) except for the isolated points \(\xu\) and \(\xel\),
  \(\norm{\tilde{x}_\text{u} - \tilde{x}_\text{l}}\) is guaranteed to be greater than \(\norm{\xu - \xel}\).
  Consequently, the disagreement function is upperbounded:
  \[
    V_-(\tilde{\chi}) = -\frac{1}{2} \sum_{\{i,j\} \in \E_-} a_{ij} \norm{\tilde{x}_\text{u} - \tilde{x}_\text{l}} <
    -\frac{1}{2} \sum_{\{i,j\} \in \E_-} a_{ij} \norm{\xu-\xel} = V_-(\chi^\ast).
  \]
  Since the disagreement function \(V_-\) decreases along the solutions of \(\dot{\chi} = -\grad V_-(\chi)\) \cite{jost08riemannian},
  which in this case is \eqref{eq:main}, the disturbed configuration \(\tilde{\chi}\) does not return to \(\chi^\ast\).
\end{proof}

\subsection{Sphere: a special case of the lemon}
\label{sec:sph}
In the special case of \(\mathcal{H}^n = \mathcal{S}^n\) when the hypermanifold is the \(n\)-sphere,
for every \(x \in \mathcal{S}^n\) and its opposite pole \(y=-x\),
\(\mathcal{I}_x\) and \(\mathcal{I}_y\) form a pair of pimples satisfying the conditions in Theorem~\ref{clm:pimballasy}.
Therefore, we have for the following subset of \(\Clem\), which is a polarization set specialized on the \(n\)-sphere
\begin{equation}
  \label{eq:csph}
  \begin{split}
    \Csph \coloneqq&
    \{\chi \in (\mathcal{S}^n)^N \,|\, x_i = x\> \forall i \in \Vu, x_i = y\> \forall i \in \Vl, \norm{x-y}_2=2\} \\
    =& \{\chi \in (\mathcal{S}^n)^N \,|\, x_i = x\> \forall i \in \Vu, x_i = -x\> \forall i \in \Vl,  x\in\mathcal{S}^n\}.
  \end{split}
\end{equation}
\begin{corollary}
  \label{cor:sph}
  For system \eqref{eq:main} evolving on \(\mathcal{S}^n\), \(\Csph\) given by \eqref{eq:csph}
  constitutes an asymptotically stable set of polarized equilibria.
\end{corollary}
Moreover, for \(n\)-spheres excepting the circle \(\mathcal{S}^1\),
a stronger result of almost global asymptotic stability can be obtained by exploiting the spherical symmetry.
By almost global asymptotic stability, we mean
\begin{definition}[almost global asymptotic stability]
  \label{def:mea0}
  A set of equilibria \(\mathcal{D} \subset (\Hn)^N\) is almost globally attractive if for all initial conditions
  except a measure-zero subset, it holds that \(\lim_{t\rightarrow\infty} \chi(t) \in \mathcal{D}\).
\end{definition}
The measure-zero set in Definition~\ref{def:mea0} is with respect to the Lebesgue measure.
For a precise definition on the \(n\)-sphere, we refer to \cite[Def.~4]{markdahl18almost}.
\begin{theorem}
  \label{thm:sph}
  For system \eqref{eq:main} evolving on \(\mathcal{S}^n\) with \(n > 1\),
  the polarization set \(\Csph\) given by \eqref{eq:csph} is almost globally asymptotically stable.
  Moreover, every trajectory converges to some point in \(\Csph\) at a locally exponential rate.
\end{theorem}
\begin{proof}
  Apply a coordinate transformation \(y_i = x_i\) for \(i \in \Vu\) and \(y_i = -x_i\) for \(i \in \Vl\).
  System \eqref{eq:main} becomes
  \begin{equation}
    \label{eq:sn}
    \dot{y}_i = (I - y_i y_i') \sum_{j\in\V} a_{ij} y_j, \quad \forall i \in \V,
  \end{equation}
  where we have used the facts that \(\normal_i = x_i\) and \((I - x_i x_i')x_i = 0\) that are valid on the unit \(n\)-spheres.
  System \eqref{eq:sn} is a special case (constant \(a_{ij}\)) of a family of consensus protocols
  considered in \cite[Thm.~13]{markdahl18almost}, which guarantees the almost global asymptotic stability of the consensus set
  \[
    \mathcal{C} = \{(y_i)_{i=1}^N \in (\mathcal{S}^n)^N \,|\, y_i = y_j, \forall i,j \in \V\}.
  \]
  This consensus set \(\mathcal{C}\) maps to the polarization set \(\Csph\) by reversing the bijective coordinate transform.
  Therefore, applying \cite[Thm.~13]{markdahl18almost} to system \eqref{eq:sn} and then reversing the coordinate transformation,
  we obtain the first conclusion.
  For the second statement, notice that the adjacency matrix \(A\) is constant and nonnegative.
  Then, \cite[Thm.~13]{markdahl18almost} again applies.
\end{proof}
Thus, Theorem \ref{thm:sph} together with Corollary \ref{cor:sph} complements and generalizes several existing results in the literature.
For instance, \cite[Thm.~3.4]{ha20emergence} provides regions of attraction with exponential stability
under a condition on relative strengths of attractive and repulsive gains, assuming all-to-all networks.
The almost glocal asymptotic stability result in \cite[Thm.~4.5]{song17intrinsic}
is only established on the 2-sphere, and is restricted to cycle graphs with even number of agents.
In comparison, we are able to conclude without extra conditions that the polarized equilibria set is
almost globally asymptotically stable for \(\mathcal{S}^n\) with \(n>1\),
and locally asymptotically stable for \(\mathcal{S}^n\) for all \(n\in\mathbb{N}\).

\section{Donuts}
\label{sec:donuts}
Observe the apparent symmetry of the two scenarios considered in Section~\ref{sec:fruits},
namely, a pair of dimples with attractive intergroup coupling vs.\ a pair of pimples with repulsive intergroup coupling.
It might come as a little surprise that
the conditions for stability do not exhibit similar symmetry in Props.~\ref{clm:dimhmin} and \ref{clm:pimballmin}.
The hidden symmetry is, however, revealed by additional geometric possibilities illustrated in Fig.~\ref{fig:doughty}.
As we shall see, these scenarios can indeed be addressed using the same conditions discovered in Section~\ref{sec:fruits}.
\begin{figure}
  \centering
  \includegraphics[width=.3\linewidth]{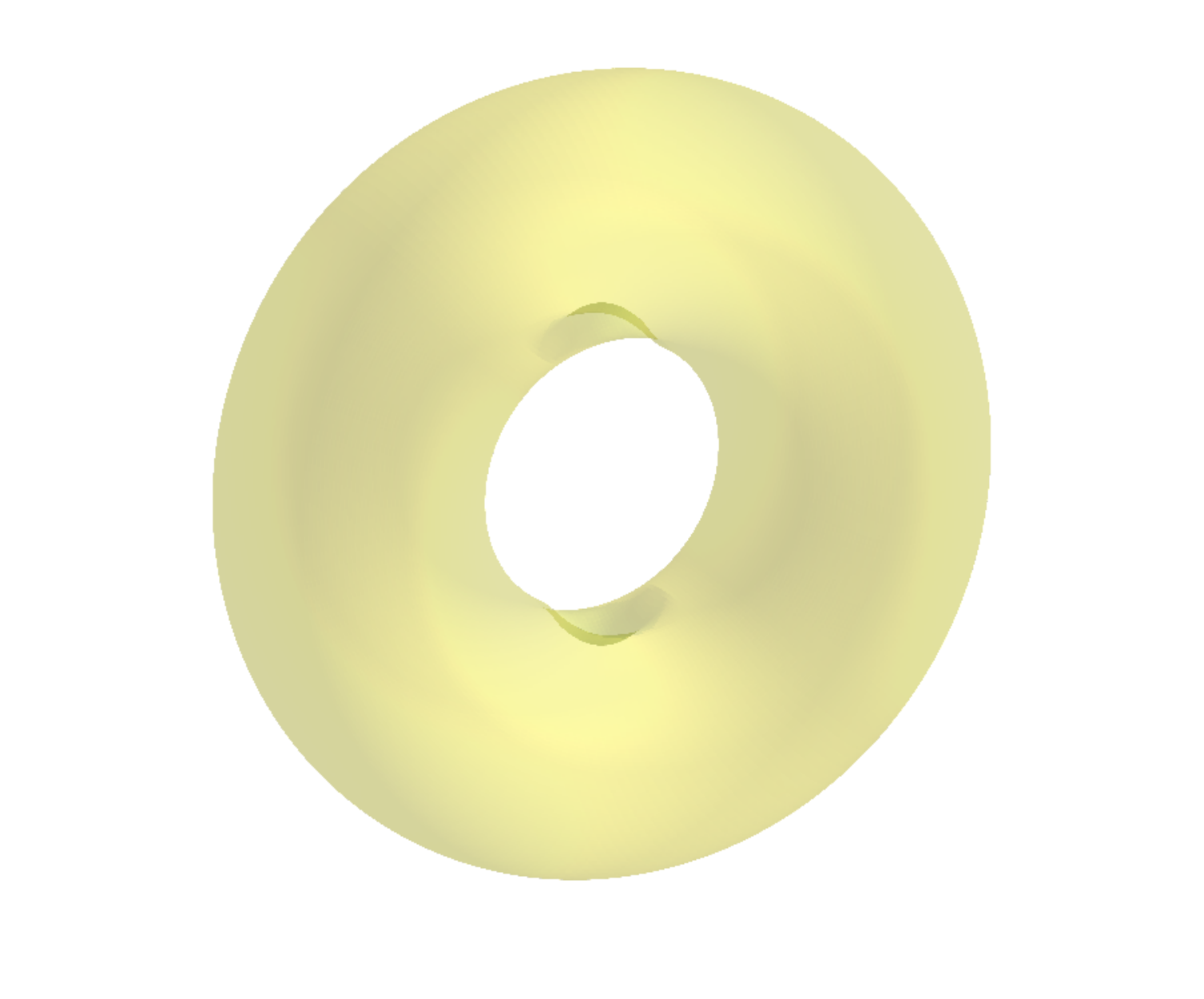}
  \includegraphics[width=.3\textwidth]{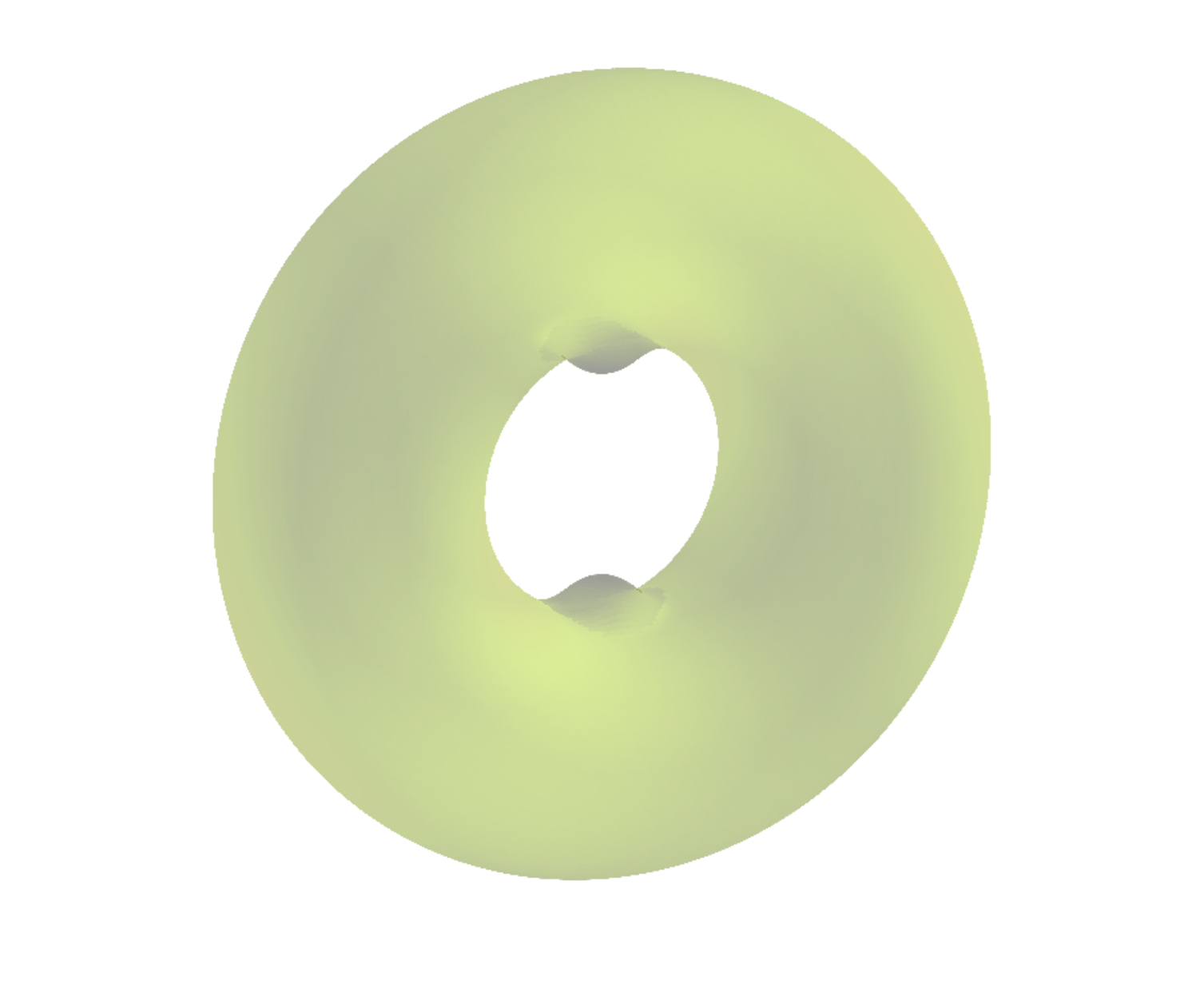}
  \includegraphics[width=.3\linewidth]{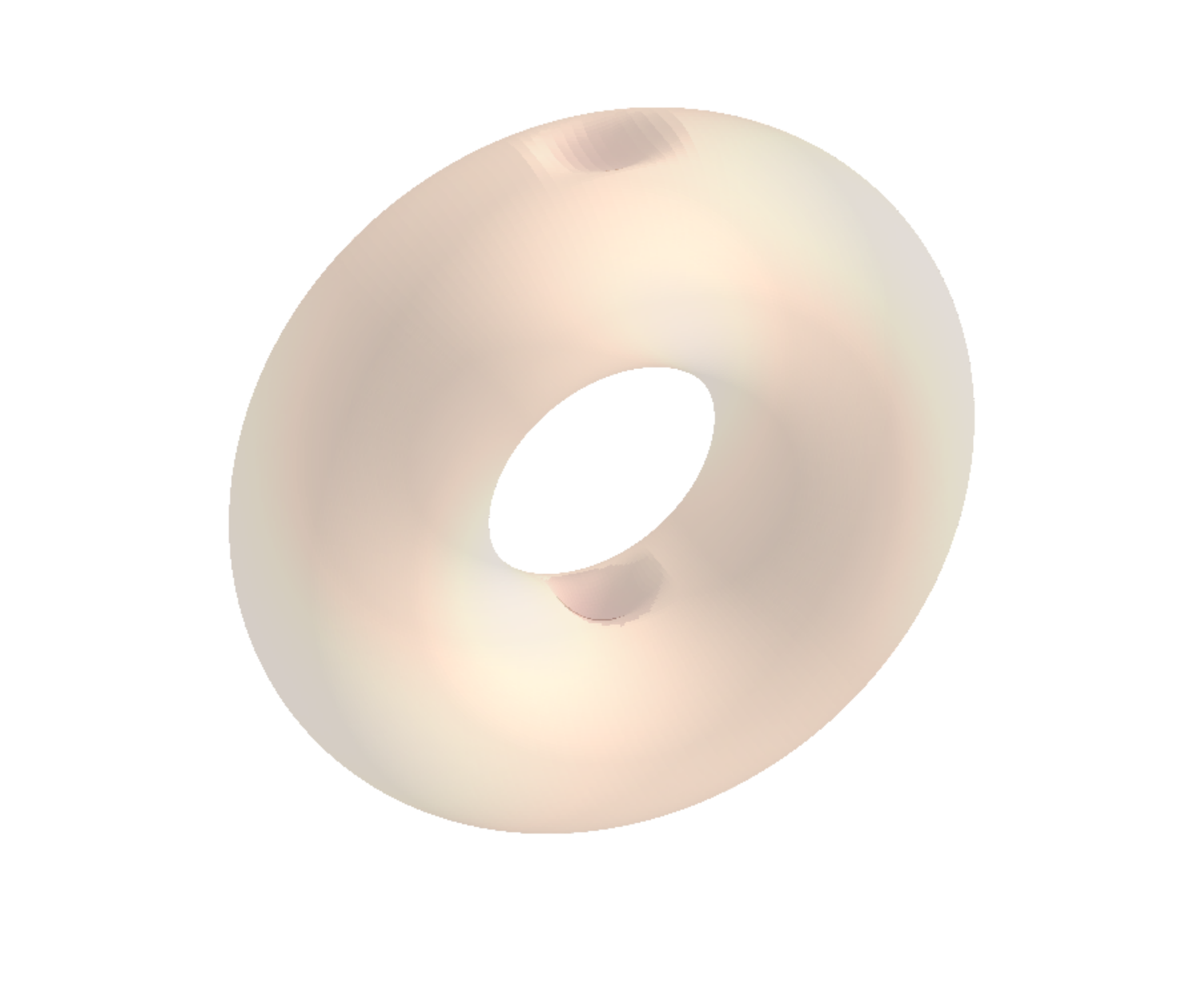}
  \caption{A lemon donut with a pair of dimples (left), an apple donut with a pair of pimples (middle),
    and a peach donut with one dimple and one pimple (right).}
  \label{fig:doughty}
\end{figure}

Stable polarization with a pair of dimples also exists in the system~\eqref{eq:main} with repulsive intergroup coupling,
if the normals of the two dimples point toward each other, see the lemon donut in Fig.~\ref{fig:doughty}.
The stability conditions in the following result is identical to that in Corollary~\ref{cor:pimballyap} and Theorem~\ref{clm:pimballasy}.
\begin{proposition}
  \label{clm:dimball}
  For system~\eqref{eq:main} under Assumption~\ref{ass:dim},
  if there exists a pair of distinct dimple bottoms \(x_\text{u} \in \Iu\) and \(x_\text{l} \in \Il\) such that
  \(\mathcal{I}_\text{u}\) and \(\mathcal{I}_\text{l}\) are entirely contained in \(\mathcal{B}_{r_o}(x_o)\),
  then \(\mathcal{C}_\text{lem}\) defined in~\eqref{eq:clem} is Lyapunov stable.
  Furthermore, if there is a neighborhood \(\nei_\text{a}(\mathcal{C}_\text{lem})\) on \((\mathcal{H}^n)^N\)
  that belongs to an analytic manifold, then \(\mathcal{C}_\text{lem}\) is asymptotically stable.
\end{proposition}
\begin{proof}
  Identical to the proof of Corollary~\ref{cor:pimballyap} and Theorem~\ref{clm:pimballasy}.
\end{proof}
\begin{remark}
  The conditions in Prop.~\ref{clm:dimhmin} and~\ref{clm:dimball} do not intersect.
\end{remark}
For the apple donut in Fig.~\ref{fig:doughty} with attractive intergroup coupling,
we have the following result analogous to that for the apple,
although the condition here has a switched inequality (compare to Prop.~\ref{clm:dimhmin}).
\begin{proposition}
  \label{clm:pimh}
  For system~\eqref{eq:mainpp} under Assumption~\ref{ass:pim},
  if there exists a pair of distinct pimple bottoms
  \(x_\text{u} \in \mathcal{I}_\text{u}\) and \(x_\text{l} \in \mathcal{I}_\text{l}\) such that
  \(x_\text{u} - x_\text{l}\) is parallel to \(\normal(x_\text{u})\),
  and \(h_\text{u}(x_\text{l})\) is a local minimum satisfying \(h_\text{u}(x_\text{l}) > h_\text{u}(x_\text{u})\),
  then \(\chi^\ast\) defined in~\eqref{eq:capp} is Lyapunov stable.
  Furthermore, if there is a neighborhood \(\nei_\text{a}(\chi^\ast)\) on \((\mathcal{H}^n)^N\) that belongs to an analytic manifold,
  then \(\chi^\ast\) is asymptotically stable.
\end{proposition}
\begin{proof}
  Identical to the proof of Corollary~\ref{cor:dimhlyap} and Theorem \ref{clm:dimhasy}.
\end{proof}
The results in §\ref{sec:fruits} and \ref{sec:donuts} have an immediate extension to products of finitely many spaces.
We briefly describe the formalism and provide an opinion dynamics interpretation.
Let \(\mathcal{H}\) be a product of \(m\) closed and orientable hypermanifolds
\(\mathcal{H} = \mathcal{H}^{n_1} \times \mathcal{H}^{n_2} \times \dotsm \times \mathcal{H}^{n_m}\).
Dimples and pimples are features considered as belonging to each \(\mathcal{H}^{n_i}\).
Every agent has \(m\) positions, each independent from the remaining \(m-1\) positions,
and co-evolve with the positions of other agents on the same hypermanifold.
Thus, \(m\) disagreement functions are simultaneously and independently minimized.
The graphs of the agent networks and the signs of the intergroup couplings need not be the same on each hypermanifold.
We do require that there exists a structurally balanced partition of the \(N\) agents
that is admitted by all the graphs \(\G_1, \G_2,\dotso, \G_m\),
so that the same two groups consistently satisfy either Assumption~\ref{ass:dim} or \ref{ass:pim} across all the hypermanifolds.
\begin{proposition}
  If for each hypermanifold \(\mathcal{H}^{n_i}\), the conditions of one of
  Theorem~\ref{clm:dimhasy}, \ref{clm:pimballasy}, \ref{clm:dimball}, or \ref{clm:pimh} are satisfied,
  then the polarized equilibria set \(\{\chi^\ast, \Clem\}^m\) over the product space \(\mathcal{H}\) is asymptotically stable.
\end{proposition}
This conceptually simple extension accommodates the need for modeling polarization over multiple topics.
Each hypermanifold may represent a distinct issue, over which individuals take their positions.
The independent evolution of dynamics over each hypermanifold addresses the observation mentioned in the opening of this article,
that often the same polarization is formed even though the multiple topics in contention lacks apparent connections.
For a simple example, we may use \(\mathcal{S}^1 \times \mathcal{S}^2\) (the Cartesian product of a unit circle and a 2-sphere)
to model a group of people trying to determine the time and location of their meeting.

\section{Numerical examples}
\label{sec:sim}
We illustrate our main results with simple numerical examples,
and explore other possible routes to polarization not covered by the main results of Section~\ref{sec:fruits}.
\subsection{Illustrative}
\label{sec:simillu}
We demonstrate the stable polarization processes on manifolds with a pair of pimples,
the first one a lemon and the second one a sphere.
There are \(N=7\) agents divided into two groups \(\Vu = \{1,2,3\}\) and \(\Vl=\{4,5,6,7\}\).
We test two network topologies with uniform weights over all edges: the first one is connected but is otherwise arbitrarily chosen;
the second one is a bipartite graph, as portrayed in Fig.~\ref{fig:A7bip}.
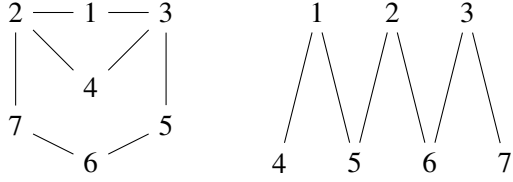
\begin{figure}
  \centering
  \begin{tikzpicture}
  \node (v1) at (0,0) {1};
  \node (v2) at (-1,0) {2};
  \node (v3) at (1,0) {3};
  \node (v4) at (0,-1) {4};
  \node (v5) at (1,-1.5) {5};
  \node (v6) at (0,-2) {6};
  \node (v7) at (-1,-1.5) {7};

  \draw (v1) -- (v2);
  \draw (v1) -- (v3);
  \draw (v2) -- (v4);
  \draw (v2) -- (v7);
  \draw (v3) -- (v4);
  \draw (v3) -- (v5);
  \draw (v5) -- (v6);
  \draw (v6) -- (v7);
\end{tikzpicture} \hspace{2em}
  \begin{tikzpicture}
  \node (v1) at (0,1) {1};
  \node (v2) at (1,1) {2};
  \node (v3) at (2,1) {3};
  \node (v4) at (-0.5,-1) {4};
  \node (v5) at (0.5,-1) {5};
  \node (v6) at (1.5,-1) {6};
  \node (v7) at (2.5,-1) {7};

  \draw (v1) -- (v4);
  \draw (v1) -- (v5);
  \draw (v2) -- (v5);
  \draw (v2) -- (v6);
  \draw (v3) -- (v6);
  \draw (v3) -- (v7);
\end{tikzpicture}
  \caption{Arbitrary and bipartite graphs used in simulations of §\ref{sec:simillu}.}
  \label{fig:A7bip}
\end{figure}
The agents update their beliefs according to the rules specified in \eqref{eq:main} with repulsive intergroup interactions.
They are randomly initialized within a pair of pimples;
in the case of the sphere, they are randomly initialized on the entire 2-sphere,
since every pair of antipodal points can serve as the pimple bottoms,
and the radii \(\epsilon\) of their neighborhoods may be very large (\(<2\)).
These specifications conform to the conditions in Theorem~\ref{clm:pimballasy},
and the trajectories depicted in Fig.~\ref{fig:simlem} converge to polarized equilibria as expected.
The accompanying convergence rates are shown in Fig.~\ref{fig:crlem}.
The log scale plots show a largely linear trend, which is evidence of exponential convergence.
\begin{figure}
  \centering
  \includegraphics[width=.3\linewidth]{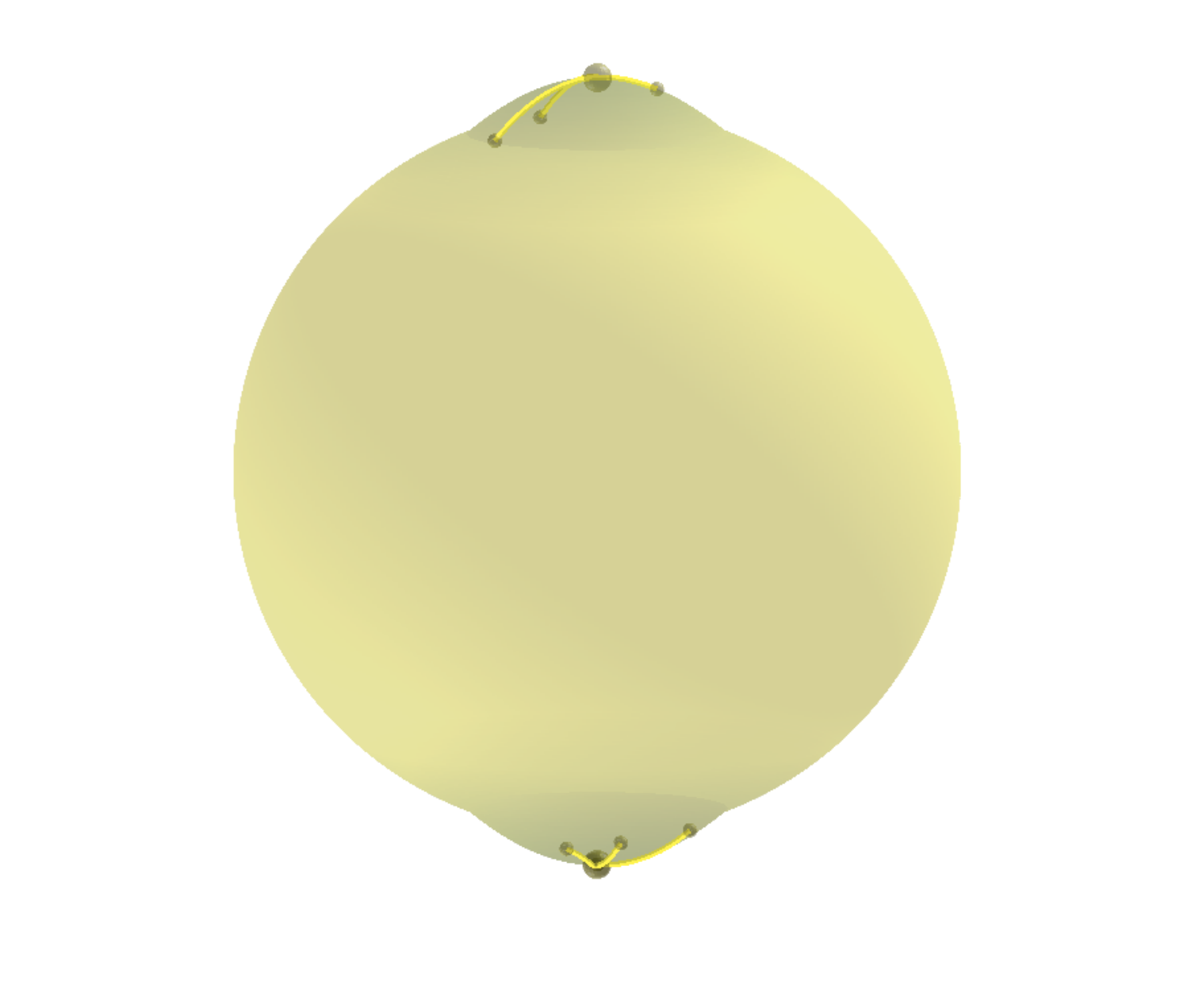}
  \includegraphics[width=.3\textwidth]{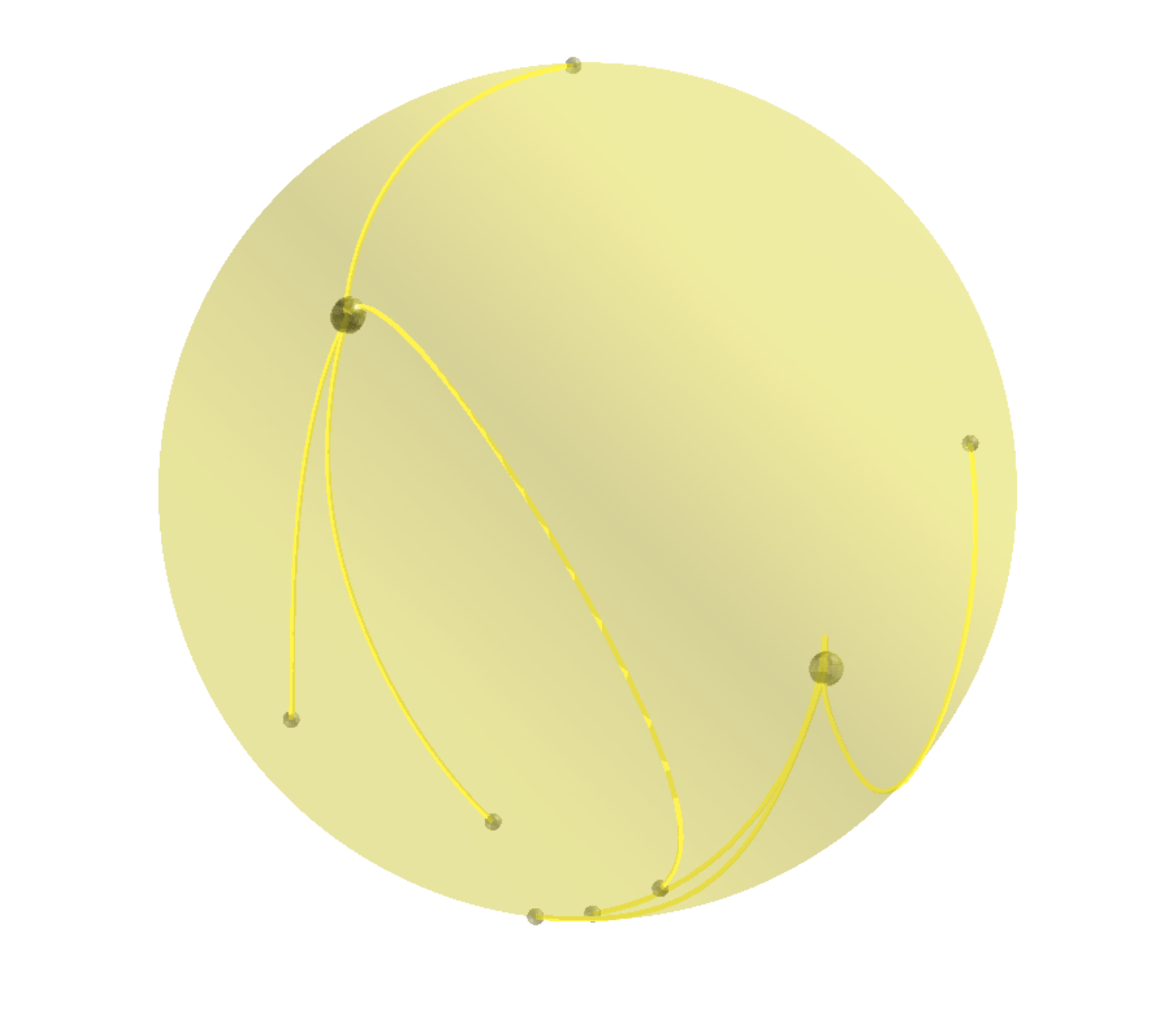}
  \includegraphics[width=.3\linewidth]{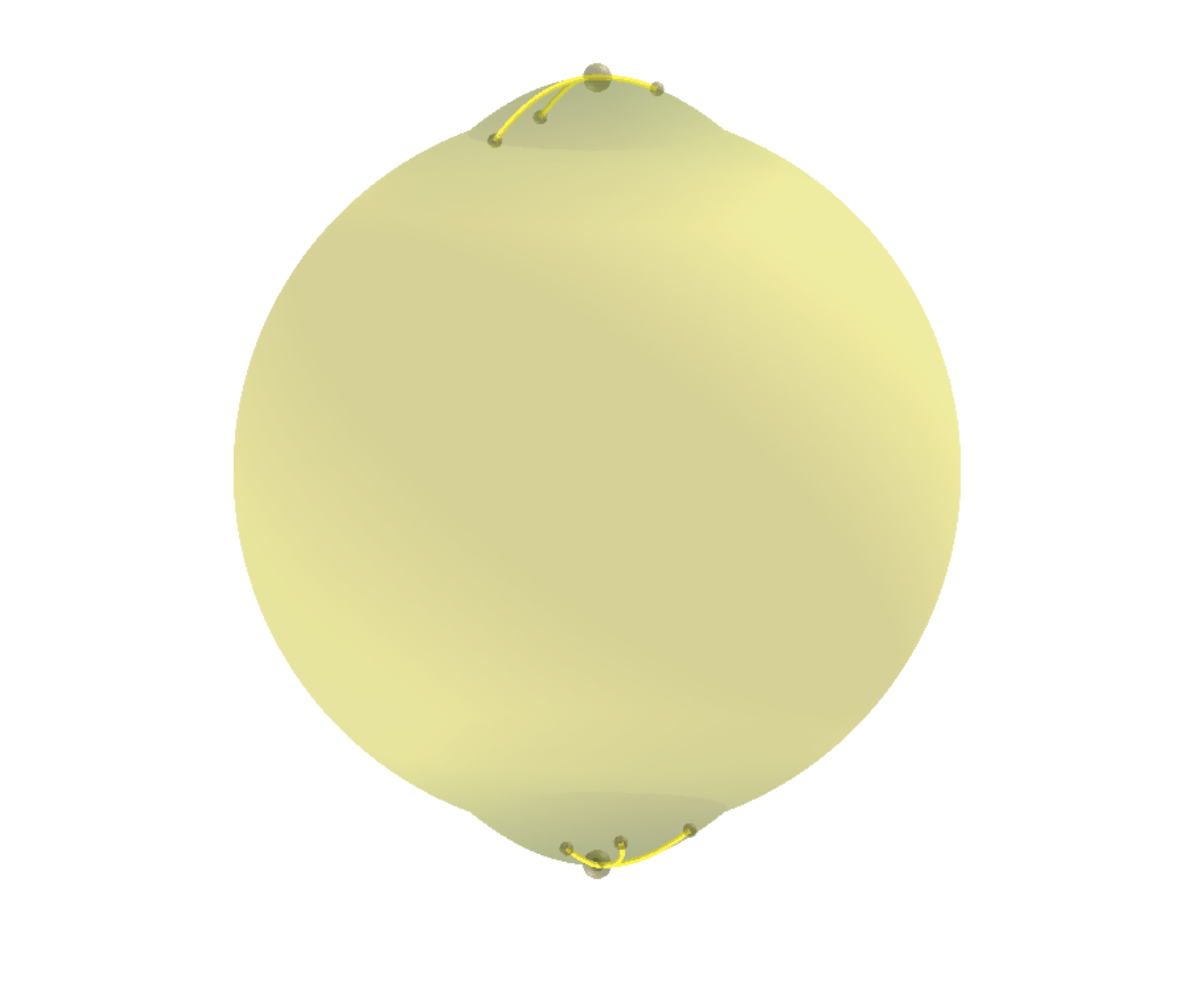}
  \caption{Two groups of agents governed by \eqref{eq:main} form stable polarization within a pair of pimples
    on a lemon with network \(A_1\) (left), on a sphere with network \(A_1\) (middle), and on a lemon with network \(A_2\) (right).
    Initial and final states are indicated by small and large dots on the yellow trajectories.}
  \label{fig:simlem}
\end{figure}
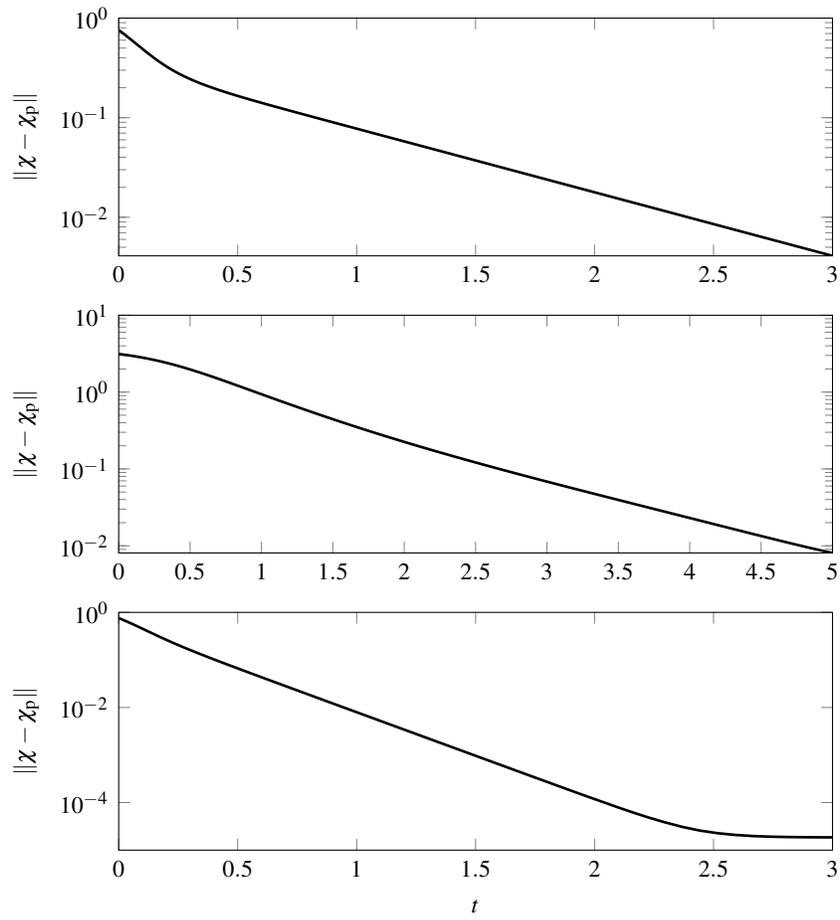
\begin{figure}
  \centering
%
%
\begin{tikzpicture}

\begin{axis}[%
width=\wdfig,
height=\htfig,
at={(0,0)},
scale only axis,
xmin=0,
xmax=3,
ymode=log,
ymin=0.00409670198003255,
ymax=1,
yminorticks=true,
ylabel={\(\norm{\chi-\chi_\text{p}}\)},
axis background/.style={fill=white},
legend style={legend cell align=left,align=left,draw=white!15!black}
]
\addplot [color=black,solid,line width=1.0pt,forget plot]
  table[row sep=crcr]{%
0	0.758525320465015\\
0.000239618078674857	0.757788324811665\\
0.000479236157349714	0.7570516464647\\
0.000718854236024571	0.75631528646141\\
0.000958472314699428	0.75557924583973\\
0.00215656270807371	0.751903869825559\\
0.003354653101448	0.748236634400206\\
0.00455274349482228	0.744577669975944\\
0.00575083388819657	0.740927107301949\\
0.011741285855068	0.722804908529813\\
0.0177317378219394	0.704912523084339\\
0.0237221897888108	0.687266475758278\\
0.0297126417556823	0.669883250932808\\
0.0501158345661635	0.612861806070449\\
0.0705190273766448	0.559691147489906\\
0.0909222201871261	0.510855759914913\\
0.111325412997607	0.466671774915488\\
0.133763812828145	0.423597551168016\\
0.156202212658682	0.386140056769004\\
0.178640612489219	0.353904728999233\\
0.201079012319757	0.326324034613018\\
0.219065959821212	0.307134292177412\\
0.237052907322667	0.290187864536915\\
0.255039854824122	0.275176660088343\\
0.273026802325577	0.261817735816661\\
0.291148237875958	0.249775287029649\\
0.309269673426339	0.238939929688067\\
0.327391108976721	0.229127710524411\\
0.345512544527102	0.220182014106032\\
0.36535719615964	0.211225549027817\\
0.385201847792178	0.20301397105701\\
0.405046499424717	0.195433950871135\\
0.424891151057255	0.18839181314451\\
0.443718207180245	0.182138577363397\\
0.462545263303236	0.176248614188073\\
0.481372319426226	0.170678169899984\\
0.500199375549217	0.165390468844024\\
0.516389568721772	0.161045940066404\\
0.532579761894328	0.156871457309705\\
0.548769955066883	0.15285284857767\\
0.564960148239438	0.148977771689533\\
0.584944883156775	0.144376571938163\\
0.604929618074111	0.139960597746595\\
0.624914352991448	0.135715169017044\\
0.644899087908784	0.131627627680836\\
0.666853994888868	0.127306187271088\\
0.688808901868952	0.123149862503533\\
0.710763808849035	0.119147987468862\\
0.732718715829119	0.115291224008396\\
0.75696339418977	0.111191014551118\\
0.781208072550421	0.10724821733398\\
0.805452750911071	0.103454341202947\\
0.829697429271722	0.099801778250888\\
0.857234073973947	0.0958160249680187\\
0.884770718676171	0.0919948212866668\\
0.912307363378395	0.0883300486013958\\
0.93984400808062	0.0848142558178794\\
0.968538393138173	0.0813017640478102\\
0.997232778195726	0.0779364776192202\\
1.02592716325328	0.0747117443330269\\
1.05462154831083	0.0716213145179538\\
1.08640034246142	0.0683483292859584\\
1.118179136612	0.0652253995217551\\
1.14995793076259	0.0622454727973861\\
1.18173672491318	0.0594018623542455\\
1.21947303549613	0.0561934849823432\\
1.25720934607909	0.0531584419800725\\
1.29494565666205	0.0502873354642139\\
1.33268196724501	0.047571265506156\\
1.37802327007244	0.0445010393016872\\
1.42336457289988	0.041628918986848\\
1.46870587572731	0.0389421637167485\\
1.51404717855475	0.0364287879846294\\
1.57021059611366	0.0335391574216832\\
1.62637401367257	0.0308787783943166\\
1.68253743123147	0.0284295539314177\\
1.73870084879038	0.0261746560592412\\
1.80706450021027	0.0236697689266371\\
1.87542815163015	0.0214047539804788\\
1.94379180305003	0.0193567850488297\\
2.01215545446992	0.017504893177086\\
2.08715545446992	0.0156762726178252\\
2.16215545446992	0.0140388521718482\\
2.23715545446992	0.0125727831817184\\
2.31215545446992	0.0112599390088225\\
2.38715545446992	0.0100840811054723\\
2.46215545446992	0.00903113469978841\\
2.53715545446992	0.00808833990627652\\
2.61215545446992	0.00724404571346931\\
2.68715545446992	0.00648781631609302\\
2.76215545446992	0.00581060580350262\\
2.83715545446992	0.00520421366791364\\
2.91215545446992	0.00466115229335978\\
2.93411659085244	0.00451313211904001\\
2.95607772723496	0.0043698159413848\\
2.97803886361748	0.00423105420853865\\
3	0.00409670198003255\\
};
\end{axis}
\end{tikzpicture}%
%
%
\begin{tikzpicture}

\begin{axis}[%
width=\wdfig,
height=\htfig,
at={(0,0)},
scale only axis,
xmin=0,
xmax=5,
ymode=log,
ymin=0.00806430785357835,
ymax=10,
yminorticks=true,
ylabel={\(\norm{\chi-\chi_\text{p}}\)},
axis background/.style={fill=white},
legend style={legend cell align=left,align=left,draw=white!15!black}
]
\addplot [color=black,solid,line width=1.0pt,forget plot]
  table[row sep=crcr]{%
0	3.13250645771525\\
0.000971058661964564	3.13090445238298\\
0.00194211732392913	3.12929910965622\\
0.00291317598589369	3.12769042879154\\
0.00388423464785825	3.12607840908806\\
0.00873952795768107	3.11796820648661\\
0.0135948212675039	3.10977445232412\\
0.0184501145773267	3.10149710791332\\
0.0233054078871495	3.09313616196293\\
0.0475818744362636	3.05007898009214\\
0.0718583409853777	3.0049461714737\\
0.0961348075344918	2.95776930716672\\
0.120411274083606	2.90859819623333\\
0.163100361429059	2.81750805788757\\
0.205789448774512	2.7209786144858\\
0.248478536119966	2.61965900267916\\
0.291167623465419	2.51432343622756\\
0.36659626471757	2.3210155269495\\
0.442024905969721	2.1235162396765\\
0.517453547221873	1.92715168754697\\
0.592882188474024	1.73687564965643\\
0.666324178600555	1.5617980991244\\
0.739766168727086	1.39940194016272\\
0.813208158853617	1.25108410044406\\
0.886650148980148	1.11719024073373\\
0.948816331061938	1.01482406147391\\
1.01098251314373	0.921923152676758\\
1.07314869522552	0.837851242655463\\
1.13531487730731	0.761900886110289\\
1.19243311567722	0.698640480048031\\
1.24955135404713	0.64108605940387\\
1.30666959241704	0.588736988006708\\
1.36378783078695	0.541116550286895\\
1.43617063644257	0.486865371804088\\
1.5085534420982	0.438694762535425\\
1.58093624775382	0.395890964331348\\
1.65331905340944	0.357804543581468\\
1.7398264101121	0.317671856758807\\
1.82633376681477	0.282641665118515\\
1.91284112351743	0.252007003271387\\
1.9993484802201	0.22514299779881\\
2.11176234678161	0.194995064926716\\
2.22417621334313	0.16940259807909\\
2.33659007990465	0.147605118445531\\
2.44900394646616	0.128944543113217\\
2.57400394646616	0.11123103962756\\
2.69900394646616	0.0962009196615773\\
2.82400394646616	0.0834020008746172\\
2.94900394646616	0.0724450900737741\\
3.07400394646616	0.0630183731925096\\
3.19900394646616	0.0548921724133105\\
3.32400394646616	0.0478683683954455\\
3.44900394646616	0.0417759631828272\\
3.57400394646616	0.0364749412977195\\
3.69900394646616	0.0318574025221995\\
3.82400394646616	0.0278300684192498\\
3.94900394646616	0.0243120575691241\\
4.07400394646616	0.0212360968893974\\
4.19900394646616	0.0185488785806071\\
4.32400394646616	0.016204436771369\\
4.44900394646616	0.0141631206882099\\
4.57400394646616	0.0123920553210154\\
4.69900394646616	0.010865192673855\\
4.82400394646616	0.00956045659464718\\
4.94900394646616	0.0084590118998401\\
4.96175295984962	0.00835744831008442\\
4.97450197323308	0.00825781992126827\\
4.98725098661654	0.00816011147513322\\
5	0.00806430785357835\\
};
\end{axis}
\end{tikzpicture}%
%
%
\begin{tikzpicture}

\begin{axis}[%
width=\wdfig,
height=\htfig,
at={(0,0)},
scale only axis,
xmin=0,
xmax=3,
xlabel={\(t\)},
ymode=log,
ymin=1e-05,
ymax=1,
yminorticks=true,
ylabel={\(\norm{\chi-\chi_\text{p}}\)},
axis background/.style={fill=white},
legend style={legend cell align=left,align=left,draw=white!15!black}
]
\addplot [color=black,solid,line width=1.0pt,forget plot]
  table[row sep=crcr]{%
0	0.758525320465015\\
0.000130323219630468	0.75805828695457\\
0.000260646439260936	0.757591400732997\\
0.000390969658891405	0.757124661951268\\
0.000521292878521873	0.756658070760271\\
0.00117290897667421	0.754327333940856\\
0.00182452507482655	0.752000309471457\\
0.0024761411729789	0.749677016111329\\
0.00312775727113124	0.747357472562102\\
0.00638583776189294	0.735816651912009\\
0.00964391825265465	0.724372347767032\\
0.0129019987434164	0.713026826258475\\
0.0161600792341781	0.701782309679864\\
0.0303359824613469	0.654084987030477\\
0.0445118856885157	0.608508046476582\\
0.0586877889156846	0.565194191727271\\
0.0728636921428534	0.524257395972713\\
0.0898524250016442	0.47844307589212\\
0.106841157860435	0.436231048404655\\
0.123829890719226	0.397607903253176\\
0.140818623578017	0.36247778213493\\
0.158695645081329	0.329099391691409\\
0.176572666584642	0.299147007847167\\
0.194449688087955	0.272328510897249\\
0.212326709591268	0.248329026256952\\
0.228562446355325	0.228708037089706\\
0.244798183119382	0.210932848531513\\
0.261033919883438	0.194802579958327\\
0.277269656647495	0.180132729858521\\
0.296764179343898	0.164215465045359\\
0.316258702040301	0.149925868341556\\
0.335753224736703	0.137058784151131\\
0.355247747433106	0.125435455622798\\
0.373983812590127	0.115295198519799\\
0.392719877747148	0.10605462212542\\
0.411455942904169	0.0976183732675568\\
0.43019200806119	0.0899024451466855\\
0.448783269417218	0.0828862868398394\\
0.467374530773245	0.0764476459685439\\
0.485965792129272	0.0705331010212843\\
0.5045570534853	0.0650946551611643\\
0.521445918238281	0.0605309910305792\\
0.538334782991262	0.0562970693755736\\
0.555223647744244	0.0523673174926306\\
0.572112512497225	0.0487182620682343\\
0.590898676182196	0.0449629988988704\\
0.609684839867167	0.0415022591119339\\
0.628471003552137	0.0383120259432288\\
0.647257167237108	0.0353702092501643\\
0.667165198143136	0.0325014293737007\\
0.687073229049164	0.0298675944993897\\
0.706981259955191	0.0274490860090985\\
0.726889290861219	0.0252278198778674\\
0.746791048063524	0.0231879105224259\\
0.766692805265829	0.0213138509409503\\
0.786594562468135	0.0195920345955603\\
0.80649631967044	0.0180098782337726\\
0.829652352544156	0.0163295273702288\\
0.852808385417873	0.0148064485356073\\
0.87596441829159	0.0134259034331172\\
0.899120451165307	0.0121743639750996\\
0.925260115980707	0.0109012566956916\\
0.951399780796108	0.0097615203023321\\
0.977539445611509	0.00874124191418159\\
1.00367911042691	0.00782773038649389\\
1.0336302248529	0.00689765561064363\\
1.06358133927888	0.0060782130974539\\
1.09353245370487	0.00535636766403798\\
1.12348356813086	0.00472031024321959\\
1.15869732870628	0.0040681395212301\\
1.1939110892817	0.00350615332182382\\
1.22912484985712	0.00302208973381999\\
1.26433861043254	0.00260490609652565\\
1.30648555309284	0.00218016162699696\\
1.34863249575315	0.0018247422915665\\
1.39077943841345	0.00152766393518305\\
1.43292638107375	0.00127902651901487\\
1.47971029705198	0.00104967760830677\\
1.52649421303021	0.000861513109626706\\
1.57327812900845	0.000707437638506247\\
1.62006204498668	0.00058100844133551\\
1.66758535375543	0.000475426384045847\\
1.71510866252418	0.000389096018899215\\
1.76263197129294	0.000318664432326512\\
1.81015528006169	0.000261082685929126\\
1.86577529210791	0.000206642840610517\\
1.92139530415414	0.000163711771886991\\
1.97701531620036	0.000130056689004195\\
2.03263532824659	0.00010358503382981\\
2.09730096717741	7.96348719526875e-05\\
2.16196660610823	6.17044213307867e-05\\
2.22663224503905	4.85549816445399e-05\\
2.29129788396988	3.89468508384881e-05\\
2.36629788396988	3.09925983242643e-05\\
2.44129788396988	2.583071167199e-05\\
2.51629788396987	2.27132415039118e-05\\
2.59129788396988	2.08652049149542e-05\\
2.66629788396988	1.9760717480115e-05\\
2.74129788396987	1.91527538848692e-05\\
2.81629788396987	1.88314043967121e-05\\
2.89129788396987	1.86581377927362e-05\\
2.91847341297741	1.86160534453381e-05\\
2.94564894198494	1.85825281897576e-05\\
2.97282447099247	1.85558352632729e-05\\
3	1.85345854598361e-05\\
};
\end{axis}
\end{tikzpicture}%
  \caption{Convergence rates in log scale corresponding to the cases in Fig.~\ref{fig:simlem} from left to right.}
  \label{fig:crlem}
\end{figure}

The examples on the lemon in Fig.~\ref{fig:simlem} left and right visually demonstrate
that our polarization model, unlike the mainstream models available in the literature as pointed out in the introduction,
is able to capture the phenomenon of radicalization, where opinions become more extreme.

\subsection{Peach}
\label{sec:peach}
There is one fundamental fruit from Fig.~\ref{fig:fruity} that is left out of the discussion so far,
namely the peach with a pimple and a dimple.
It is tricky to analyze the polarization conditions for the gradient flow dynamics on a peach,
therefore we only provide numerical examples to show the possibility.
Figure~\ref{fig:simpea} depicts the same two groups of agents with adjacency matrix \(A_1\) as in §\ref{sec:simillu}.
However, the system is not governed by \eqref{eq:main} or \eqref{eq:mainpp}.
Instead, the agents in \(\Vu\) are set to be repulsed by those in \(\Vl\),
whereas the agents in \(\Vl\) are attracted to those in \(\Vu\).
Thus, the polarized equilibrium we numerically found is not covered by our theoretical results,
and we do not make any claim on its stability.
In fact, there is no guarantee that systems with such asymmetric intergroup interactions converge to an equilibrium,
as shown in Fig.~\ref{fig:simHairy}.
\begin{figure}
  \centering
  \includegraphics[width=.3\textwidth]{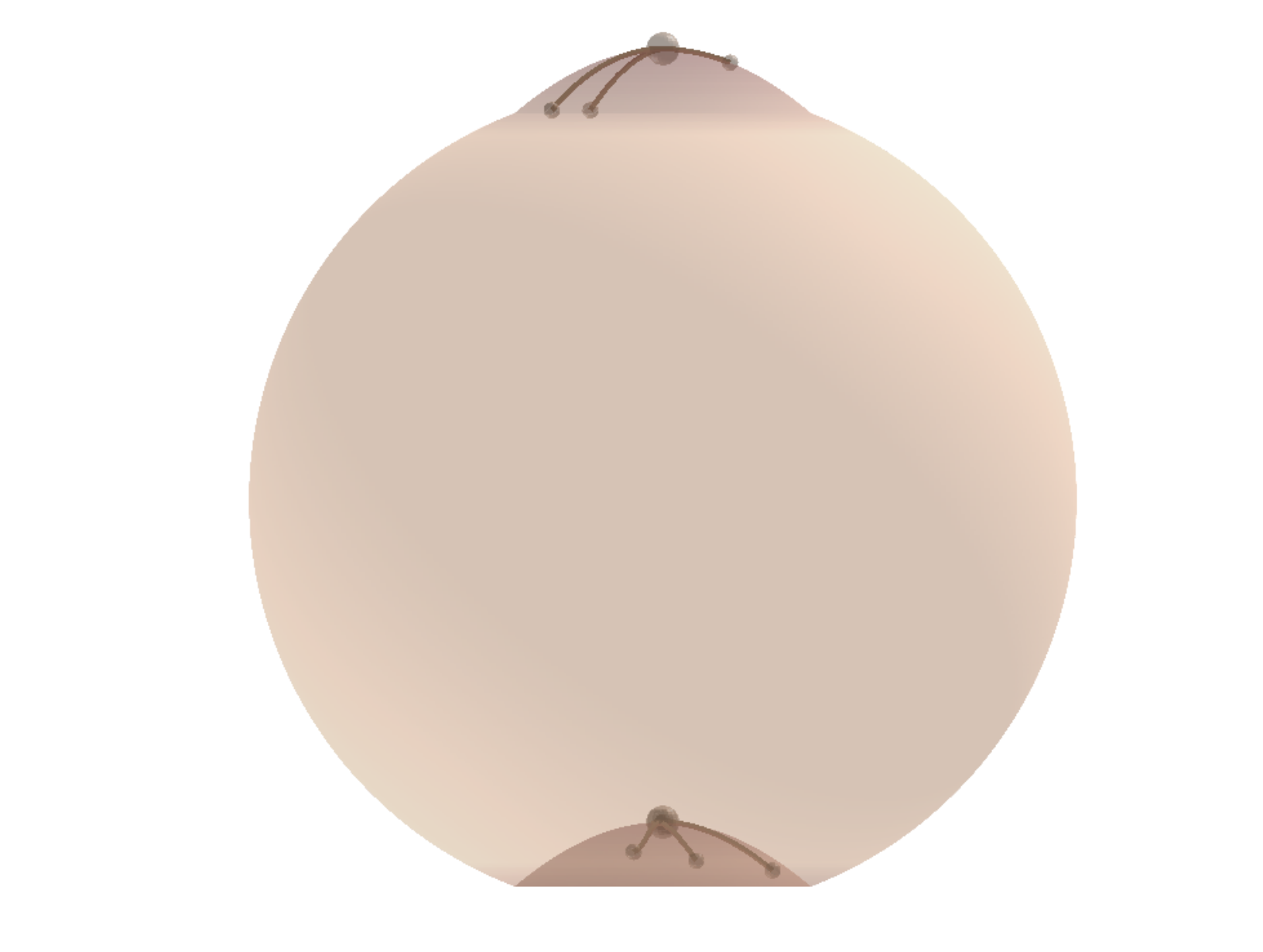}
  \caption{Two groups of agents converge to a polarized equilibrium within a pair of pimple and dimple on a peach with network \(A_1\).
    The intragroup couplings are attractive.
    The upper group is repulsed by the lower group, whereas the lower group is attracted to the upper group.
    Initial and final states are indicated by small and large dots on the pink trajectories.}
  \label{fig:simpea}
\end{figure}
\begin{figure}
  \centering
  \includegraphics[width=.33\textwidth]{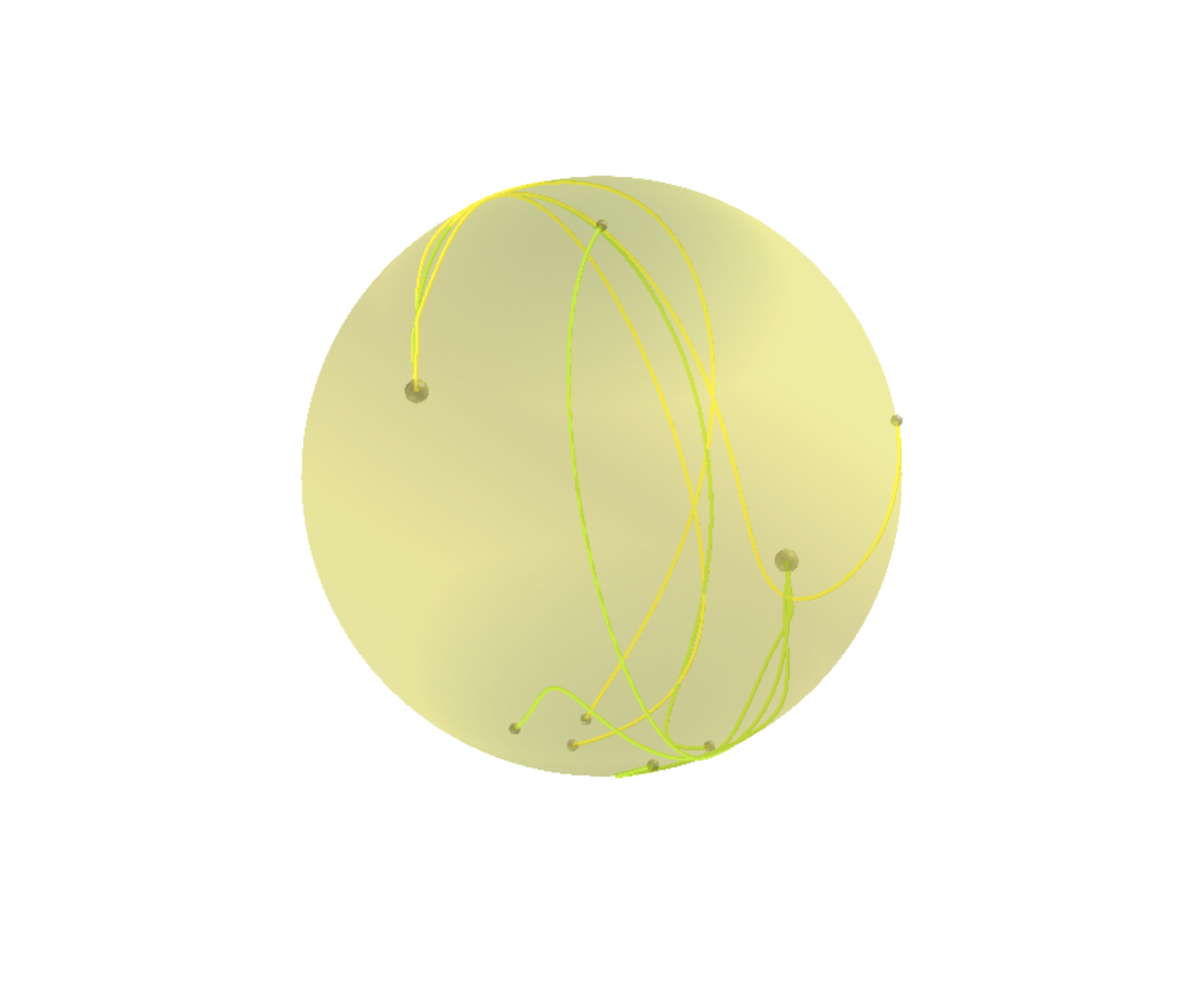}
  \includegraphics[width=.33\linewidth]{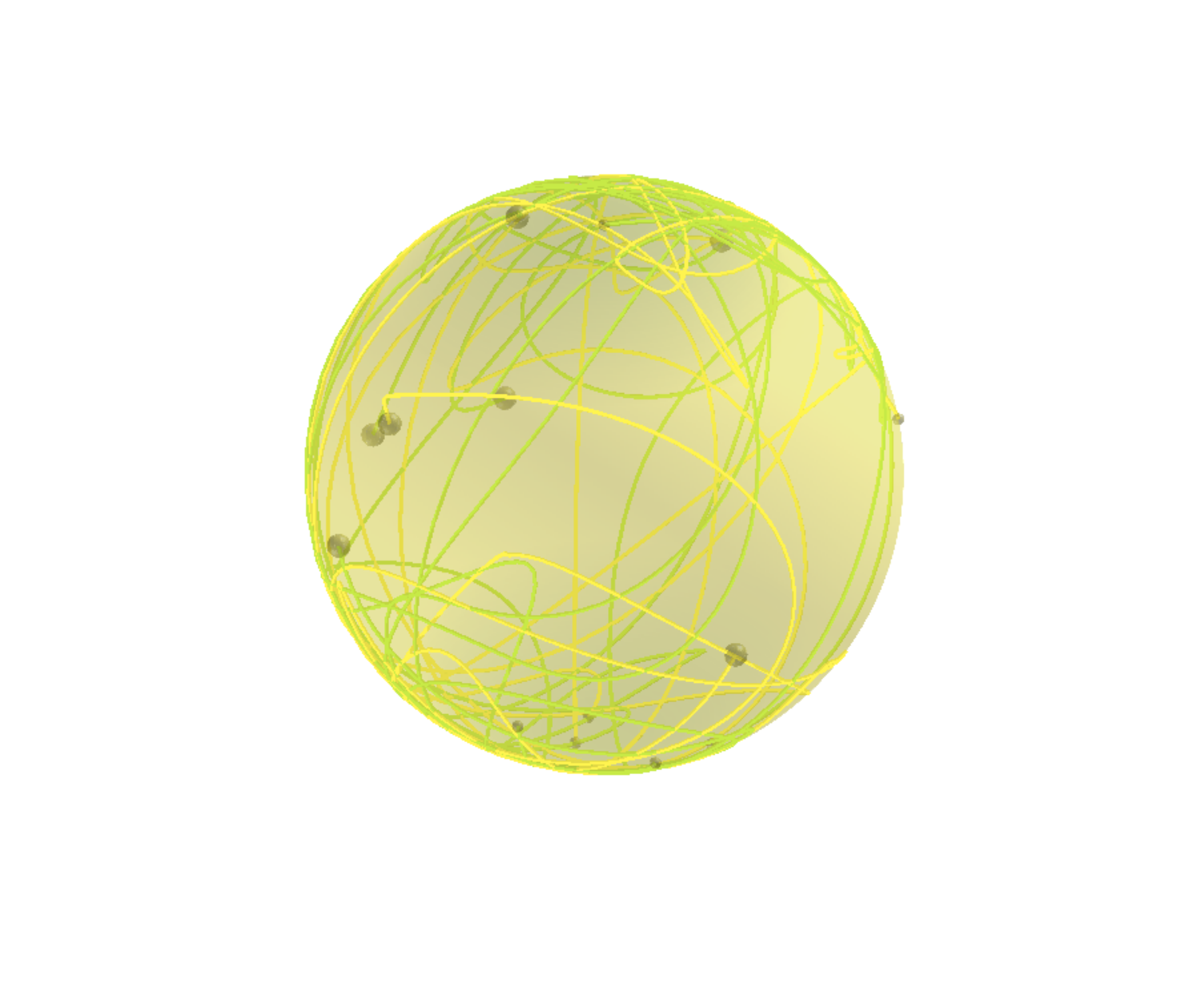}
  \caption{Asymmetric intergroup interactions:
    the upper group repulsed by the lower group and the lower group attracted to the upper group.
    The intragroup interaction remains attractive.
    Left: two groups of agents with asymmetric intergroup interactions converge to polarized equilibrium on a sphere with network \(A_1\).
    Right: the same groups of agents with asymmetric intergroup interactions
    do not converge to equilibrium on a sphere with network \(A_2\).
    Initial and final states are indicated by small and large dots on the yellow (\(\Vu\)) and green (\(\Vl\)) trajectories.}
  \label{fig:simHairy}
\end{figure}

\subsection{Weakly anomalous intergroup coupling}
\label{sec:simexpl}
One can imagine that for system \eqref{eq:main} living on a lemon,
introducing weakly attractive intergroup coupling to a small number of edges in \(\E_-\)
does not necessarily compromise polarization or its stability properties.
In turn, the tolerance for such deviation from purely repulsive intergroup coupling may be viewed as a measure of robustness.
Along this line and based on the example in §\ref{sec:simillu},
we change the weight on each edge in \(\E_-\) successively and see to what degree can polarization be maintained,
given different manifolds and initial conditions.
Each experiment is allowed sufficient time  to run for agents starting at different initial conditions to evolve.
A weight is accepted if the order parameters \(\norm{\rho_\text{u}}\) and \(\norm{\rho_\text{l}}\) are equal to 1
within the final integration time 200, where \(\rho_\text{u} = (x_1+\dotsb+x_M)/M\) and similarly for \(\rho_\text{l}\).
Complete phase synchronization is ruled out by visual inspection.

Table~\ref{tab:robsym} summarizes the findings.
\begin{table}
  \centering
  \begin{tabular}{ccc}
    \(a_{24}\) and \(a_{34}\) & 0.83 & 0.41 \\
    \(a_{27}\) and \(a_{35}\) & 0.099 & 0.23
  \end{tabular}
  \caption{Maximal weights for each edge in \(\E_-\) when changed to attractive coupling are the same for 7 different initial conditions,
    tested on the lemon (left) and the sphere (right) with the graph structure in Fig.~\ref{fig:A7}(left).}
  \label{tab:robsym}
\end{table}
Notice that \(a_{24} = a_{34}\) and \(a_{27} = a_{35}\) because, as it turns out, the network has a symmetric structure.
The reason that the initial conditions appear to have no effect on the attainable weights is perhaps less obvious.
However, assuming that no new equilibrium points are created within the dimples by the change of sign on the link,
the only option is to converge to the known equilibrium point, irrespective of the initial conditions.
As for why the edge \(\{2,4\}\) appears to be more robust to treason than \(\{2,7\}\) is harder to explain.
It is tempting to link the phenomenon to graph properties such as node degrees,
but it might also be altogether a spurious phenomenon.
\begin{figure}
  \centering
  \begin{tikzpicture}
  \node (v1) at (0,0) {1};
  \node (v2) at (-1,0) {2};
  \node (v3) at (1,0) {3};
  \node (v4) at (0,-1) {4};
  \node (v5) at (1,-1.5) {5};
  \node (v6) at (0,-2) {6};
  \node (v7) at (-1,-1.5) {7};

  \draw (v1) -- (v2);
  \draw (v1) -- (v3);
  \draw (v2) -- (v4);
  \draw (v2) -- (v7);
  \draw (v3) -- (v4);
  \draw (v3) -- (v5);
  \draw (v5) -- (v6);
  \draw (v6) -- (v7);
\end{tikzpicture} \hspace{2em}
  \begin{tikzpicture}
  \node (v1) at (0,0) {1};
  \node (v2) at (-1,0) {2};
  \node (v3) at (1,0) {3};
  \node (v4) at (0,-1) {4};
  \node (v5) at (1,-1.5) {5};
  \node (v6) at (0,-2) {6};
  \node (v7) at (-1,-1.5) {7};

  \draw (v1) -- (v2);
  \draw (v1) -- (v3);
  \draw (v2) -- (v4);
  \draw (v2) -- (v7);
  \draw (v3) -- (v4);
  \draw (v3) -- (v5);
  \draw (v4) -- (v7);
  \draw (v5) -- (v6);
\end{tikzpicture}
  \caption{Symmetric and asymmetric networks}
  \label{fig:A7}
\end{figure}
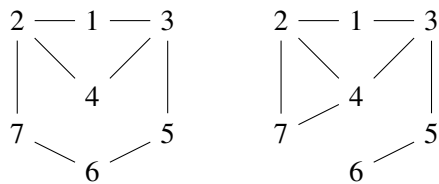

Table~\ref{tab:robasym} compiles the edge robustness for an asymmetric network, obtained by switching the link \(\{6,7\}\) to \(\{4,7\}\),
leaving the group membership unchanged.
The N/A values mean that the system fail to completely polarize.
There is not much to say quantitatively about the actual robustness values,
as the outcome heavily depends on the geometry of the underlying manifold.
\begin{table}
  \centering
  \begin{tabular}{ccc}
    \(a_{24}\) & 0.2 & 0.56 \\
    \(a_{27}\) & 0.94 & 0.82 \\
    \(a_{34}\) & 0.9 & 0.35 \\
    \(a_{35}\) & N/A & N/A
  \end{tabular}
  \caption{Maximal weights for each edge in \(\E_-\) when changed to attractive coupling are the same for 7 different initial conditions,
    tested on the lemon (left) and the sphere (right) with the graph structure in Fig.~\ref{fig:A7}(right).}
  \label{tab:robasym}
\end{table}

\section{Conclusions}
\label{sec:close}
We have established sufficient conditions for asymptotic stability of polarized equilibria
of arbitrarily connected multi-agent gradient flow systems over nonlinear spaces.
These sufficient conditions are tailored to four scenarios likely to be found on nonflat hypermanifolds,
arising from the combination of dimple/pimple pairs and attractive/repulsive intergroup couplings.
In particular, the hypersphere as a special case of general manifolds provides effortless generalization of previously known results.
Its highly symmetric nature further allows us to prove almost global asymptotic stability except for the unit circle.
A natural next step is to work with manifolds of dimensions much lower than the ambient space,
whereby the normal space dimension exceeds 1.
This would strengthen the proposed interpretation of low dimensional core belief space in the high dimensional external space.

\bibliographystyle{ieeetr}
\bibliography{bilibili}

\end{document}